\documentclass{article}

\PassOptionsToPackage{numbers, compress}{natbib}
\usepackage[preprint]{neurips_2024}
\usepackage[utf8]{inputenc} 
\usepackage[T1]{fontenc}    
\usepackage{graphicx}
\usepackage{amsmath,amssymb,amsthm,mathabx}
\usepackage{booktabs}
\usepackage{algorithmic}
\usepackage[linesnumbered,ruled,vlined]{algorithm2e}
\usepackage{acronym}
\usepackage{enumitem}
\usepackage[
    colorlinks, 
    linkcolor=green,
    anchorcolor=blue, 
    citecolor=blue
]{hyperref}
\usepackage{setspace}
\usepackage{subfigure}
\usepackage[skip=3pt,font=small]{subcaption}
\usepackage[skip=3pt,font=small]{caption}
\usepackage[dvipsnames]{xcolor}
\usepackage[capitalise]{cleveref}
\usepackage{tabularx,colortbl,multirow,array,makecell}
\usepackage{overpic,wrapfig}

\usepackage{listings}
\lstdefinestyle{mystyle}{
  backgroundcolor=\color{backcolour},   commentstyle=\color{codegreen},
  keywordstyle=\color{magenta},
  numberstyle=\tiny\color{codegray},
  stringstyle=\color{codepurple},
  basicstyle=\ttfamily\footnotesize,
  commentstyle=\color{red!10!green!70}\textit,
  breakatwhitespace=false,         
  breaklines=true,                 
  captionpos=b,                    
  keepspaces=true,                 
  numbers=left,                    
  numbersep=5pt,                  
  showspaces=false,                
  showstringspaces=false,
  showtabs=false,                  
  tabsize=2
}

\makeatletter
\DeclareRobustCommand\onedot{\futurelet\@let@token\@onedot}
\def\@onedot{\ifx\@let@token.\else.\null\fi\xspace}
\def\eg{\emph{e.g}\onedot} 

\def\ie{\emph{i.e}\onedot}

\makeatother

\DeclareMathOperator*{\argmax}{arg\,max}
\DeclareMathOperator*{\argmin}{arg\,min}

\newtheorem{theorem}{Theorem}

\newtheorem{lemma}{Lemma}
\newtheorem{remark}{Remark}
\makeatletter 

\theoremstyle{definition}

\makeatletter
\renewcommand{\paragraph}{%
  \@startsection{paragraph}{4}%
  {\z@}{0ex \@plus 0ex \@minus 0ex}{-1em}%
  {\hskip\parindent\normalfont\normalsize\bfseries}%
}
\makeatother

\graphicspath{{figures/}}

\crefname{algocf}{alg.}{algs.}
\Crefname{algocf}{Algrithm}{Algrithm}

\definecolor{gblue}{HTML}{4285F4}
\definecolor{gred}{HTML}{DB4437}
\definecolor{gray}{gray}{0.9}

\acrodef{drc}[DRC]{Deep Region Competition}
\acrodef{moe}[MoE]{Mixture of Experts}
\acrodef{mlem}[ML-EM]{Maximum-Likelihood Expectation-Maximization}
\acrodef{lebm}[LEBM]{Latent-space Energy-Based Model}
\acrodef{mle}[MLE]{Maximum Likelihood Estimation}
\acrodef{em}[EM]{Expectation-Maximization}
\acrodef{birds}[Birds]{Caltech-UCSD Birds-200-2011}
\acrodef{dogs}[Dogs]{Stanford Dogs}
\acrodef{cars}[Cars]{Stanford Cars}
\acrodef{tmds}[TM-dSprites]{Textured Multi-dSprites}
\acrodef{tvn}[TV-norm]{Total Variation norm}

\newtheorem{corollary}{Corollary}[theorem] 

\title{Watermarking Generative Tabular Data}

\author{%
  Hengzhi He\thanks{Equal Contribution.} \\
  \texttt{hengzhihe@g.ucla.edu} 
  \And
  Peiyu Yu* \\
  \texttt{yupeiyu98@g.ucla.edu}
  \And
  Junpeng Ren \\
  \texttt{junren18@g.ucla.edu}
  \And
  Ying Nian Wu \\
  \texttt{ywu@stat.ucla.edu}
  \And
  Guang Cheng \\
  \texttt{guangcheng@ucla.edu}\\
}

\begin{document}

\maketitle

\begin{abstract}
In this paper, we introduce a simple yet effective tabular data watermarking mechanism with statistical guarantees. We show theoretically that the proposed watermark can be effectively detected, while faithfully preserving the data fidelity, and also demonstrates appealing robustness against additive noise attack. The general idea is to achieve the watermarking through a strategic embedding based on simple data binning. Specifically, it divides the feature's value range into finely segmented intervals and embeds watermarks into selected ``green list" intervals. To detect the watermarks, we develop a principled statistical hypothesis-testing framework with minimal assumptions: it remains valid as long as the underlying data distribution has a continuous density function. The watermarking efficacy is demonstrated through rigorous theoretical analysis and empirical validation, highlighting its utility in enhancing the security of synthetic and real-world datasets. 
\end{abstract}

\section{Introduction}
\label{sec:intro}
The recent surge of powerful generative models has lead to increasingly adept generative data synthesizers \cite{vaswani2017attention,ho2020denoising,song2020score,yu2021unsupervised,rombach2022high,yu2022latent,qian2023learning,yu2024learning,zhang2024object} that closely mimic real datasets. 
However, the surge in AI-driven data synthesis also raises significant concerns. Distinguishing AI-generated content from human-generated content poses challenges that impact copyright infringement, privacy breaches, and the spread of misinformation. These concerns have prompted regulatory responses at both national and international levels. For example, the White House’s Executive Order \footnote{
\href{https://www.whitehouse.gov/briefing-room/presidential-actions/2023/10/30/executive-order-on-the-safe-secure-and-trustworthy-development-and-use-of-artificial-intelligence/}{https://www.whitehouse.gov/briefing-room/presidential-actions/2023/10/30/executive-order-on-the-safe-secure-and-trustworthy-development-and-use-of-artificial-intelligence/}} and the EU’s Artificial Intelligence Act \footnote{
\href{https://artificialintelligenceact.eu/wp-content/uploads/2024/02/AIA-Trilogue-Committee.pdf}{https://artificialintelligenceact.eu/wp-content/uploads/2024/02/AIA-Trilogue-Committee.pdf}} both emphasize the importance of secure, responsible AI practices and making AI-generated content detectable and traceable to uphold transparency and protect users' rights.

In the context of ensuring the integrity and authenticity of generative products, watermarking techniques emerge as a common solution. Significant advancements have been achieved in the watermarking of unstructured generative data, such as texts \cite{kirchenbauer2023watermark}, \cite{zhao2023provable} and images \cite{wen2023treerings} (please see \cref{appx:related_works} for an extended discussion of related works). However, the structured domain of tabular data remains less explored. Effective watermarking in this area must address the specific challenges of maintaining data fidelity and usability in structured datasets, which are critical in applications like healthcare and finance where data integrity is paramount. 

To fill in this important missing part on the landscape of watermarking generative data, in this work we propose, to the best of our knowledge, the first tabular data watermarking framework with solid theoretical foundation. We focus extensively on watermarking continuous variables in the tabular data.
Our proposed mechanism is achieved through a strategic embedding of watermarks using data binning. Specifically, it divides the feature's value range into finely segmented intervals and embeds watermarks into selected ``green list'' intervals. This specification of ``green list'' intervals shares the same spirit as ``green list'' techniques as in text data watermarks \citep{kirchenbauer2023watermark}, while the methodology and underlying theoretical framework are completely new. To detect the watermarks, we develop a statistical hypothesis-testing framework with minimal assumptions, requiring that the underlying data distribution has a continuous density function. Finally, we provide empirical evidence that demonstrates the effectiveness of our proposed framework on both synthetic and real-world tabular datasets. We summarize and highlight our major contributions as follows:

\begin{figure}[!tbp]
    \centering    
    \includegraphics[width=0.75\textwidth]{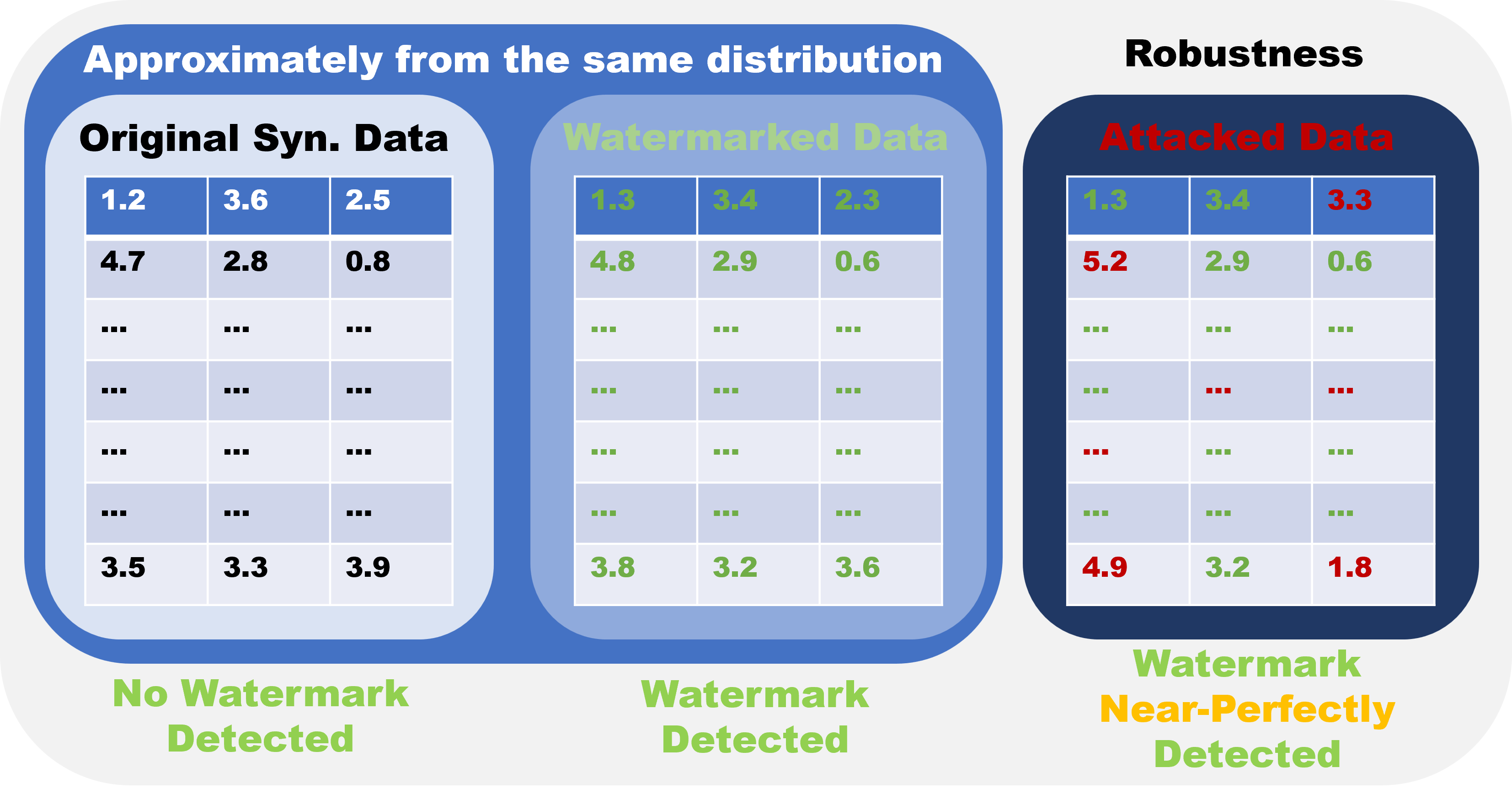}
    \caption{\textbf{Overview of the tabular data watermarking scheme.}}
    \label{fig:wm_overview}
\end{figure}

\begin{itemize}
     \item \textbf{Theoretical guarantee of data fidelity:} We show theoretically that finer or smaller intervals result in watermarked data closer to the original data, specifically with an error rate of $O(\frac{1}{m})$, where $m$ is the number of ``green list'' intervals. Empirically, we observe minimal fidelity and utility loss on both synthetic and real datasets when applying our proposed watermark to the generative tabular data.
     \item \textbf{Principled detection framework:} We propose a principled hypothesis-testing framework for tabular data watermark detection. Our testing process is backed by an interesting theoretical result that as the number of intervals $m \to \infty$, the probability of a data point falls within the ``green list'' intervals converges to $\frac{1}{2}$.
     \item \textbf{Robustness against noise masking attack:} We demonstrate appealing robustness of our proposed tabular data watermark against attacks with additive noise, where attackers use continuous noise to perturb the watermarked tabular data. Our theoretical result indicates that if the success probability of attacking an individual element is capped at \( \frac{1}{2} \),  
     then even attacking almost all elements is insufficient to significantly increase the likelihood of overcoming the hypothesis test. 
     We show that our watermark remains valid even when $\sim 95\%$ of the elements are attacked with large noise. We validate our result on both synthetic and real datasets, and observe that our watermark can be effectively detected.
   \end{itemize}

\section{Watermarking Tabular Data}
\subsection{Problem Statement} 
\label{sec:problem_state}
We consider a dataset $\mathbf{X}$, structured as an $n \times p$ table where each of the $p$ columns consists of $n$ i.i.d. data points 
from a distribution $F_i$, each with a continuous probability density function $f_i,i=1,...,p$. Typically, $\mathbf{X}$ represents synthetic data generated from some generative model, which we refer to as generative tabular data throughout the paper. Our objective is to construct a watermarked version of this dataset, denoted as $\mathbf{X}_w$. This watermarked dataset aims to achieve three primary goals (\cref{fig:wm_overview}): i) maintaining a minimal discrepancy $| \mathbf{X} - \mathbf{X}_w |$ under standard assumptions; ii) ensuring that $\mathbf{X}_w$ can be reliably identified as the outcome of our specific detection process; and iii) achieving desirable robustness against potential attacks. We next show how this watermark can be achieved with a surprisingly simple yet effective procedure. 

\subsection{Watermarking Tabular Data with Data Binning}

The proposed watermark is applied element-wise, with the detailed procedure consisting of the following steps (illustrated in \cref{fig:wm_workflow}):

\begin{algorithm}[!htbp]
\caption{\textbf{Tabular data watermarking algorithm.}}
\label{alg:wm_algo}
\KwIn{number of ``green list'' intervals $m$; 
original tabular dataset $\mathbf{X}$.}
\KwOut{
watermarked dataset $\mathbf{X}_w$.
}
\BlankLine	
\textbf{Interval Division:} The continuous interval from 0 to 1 is divided into \(2m\) equal parts, forming intervals such as \([0, \frac{1}{2m}]\), \([\frac{1}{2m}, \frac{2}{2m}]\), \ldots, \([\frac{2m-1}{2m}, 1]\), and form $m$ pairs of consecutive intervals \(\{[0,\frac{1}{2m}],[\frac{1}{2m},\frac{2}{2m}] \}, \{[\frac{2}{2m},\frac{3}{2m}],[\frac{3}{2m},\frac{4}{2m}]\},\cdots, \{[\frac{2m-2}{2m},\frac{2m-1}{2m}],[\frac{2m-1}{2m},1] \}\). 
    
\textbf{Green List Selection:} From each pair of intervals, one interval is selected as a ``green list'' interval. The resulting set of $m$ intervals is denoted as $G$, the set of ``green list'' intervals. 
    
\textbf{Watermark Embedding:}
\For{each element $x$ in $\mathbf{X}$}{

    \textbf{Finding the nearest green list interval $g$:} For the fractional part of each data point $x$, we identify the closest interval on the green list as 
    $g = 
    \argmax\limits_{{g \in G}}
    \|\left(x - i(x) \right)  - {\rm center}(g) \|$, where $i(x)$ is the integer part of $x$.
    
    \textbf{Imposing green list constraints:}
    \textbf{If} $x - i( x ) \in g$, then $x$ is left as is. \textbf{Else}, we replace $x$ as $i( x ) + r$, where $r$ is uniformly sampled from $g$.
}
\end{algorithm}

\begin{figure}[!tbp]
    \centering    
    \includegraphics[width=0.85\textwidth]{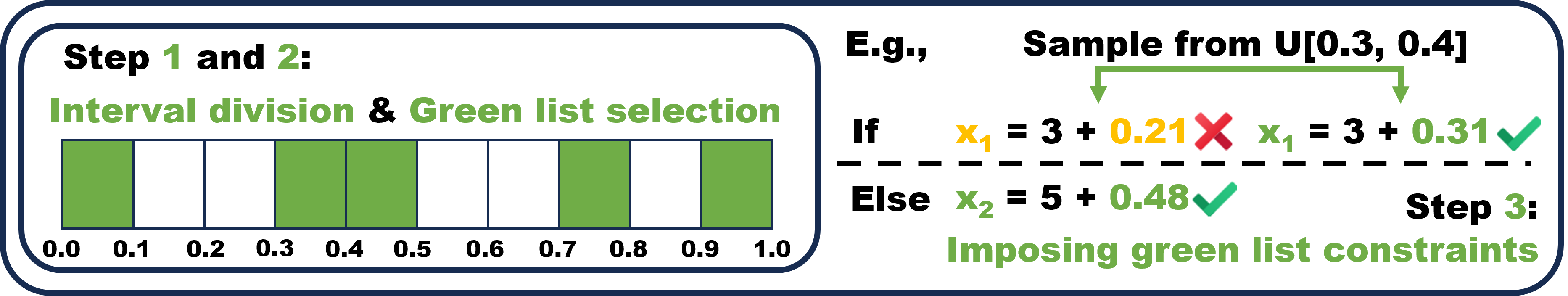}
    \caption{\textbf{Illustration of our proposed watermarking scheme for tabular data.} Specifically, our scheme consists of three major steps: i) dividing the continuous interval $[0, 1]$ into $2m$ equal parts, forming $m$ pairs of consecutive intervals; ii) randomly selecting one interval from each pair, resulting in the set of $m$ ``green list'' intervals; and iii) sampling new fractional part for the input element from the nearest ``green list'' interval if the original fractional part falls outside of this interval.}
    \label{fig:wm_workflow}
\end{figure}

\paragraph{An illustrative example}
We can consider an artificial $10 \times 1$ dataset as an illustrative example for the implementation of this watermarking process. Without loss of generality, we can assume that the value range of this dataset is 
$[0,1]$. Otherwise, we can subtract the integer part of each element in the table to obtain the fractional part that falls with in $[0, 1]$. In the context of this manuscript, we therefore consider an element $x$ as falling into the ``green list'' if and only if its fractional part $x-i(x)$ falls into one of the intervals inside the ``green list'' intervals in $[0,1]$; $i(x)$ is the integer part of $x$. For simplicity, we set the watermark with $m=5$. The value range is then divided into 10 smaller intervals $\{[0, 0.1], ..., [0.9, 1]\}$, forming 5 pairs of consecutive intervals $\{P_i \}_{i=1}^{5}$. From each pair, one interval is randomly selected for the green list. As an example, the green list intervals (highlighted in {\color{teal}green}) selected might be:
\begin{equation*}
\begin{aligned}
& P_1 = \{{\color{teal}[0.0,0.1]},[0.1,0.2]\}, 
  P_2 = \{[0.2,0.3],{\color{teal}[0.3,0.4]} \}, 
  P_3 = \{{\color{teal}[0.4,0.5]},[0.5,0.6] \}, \\ 
& P_4 = \{[0.6,0.7],{\color{teal}[0.7,0.8]} \}, 
  P_5 = \{[0.8,0.9],{\color{teal}[0.9,1.0]} \}.
\end{aligned}
\end{equation*}
Given this setup, a data point $x$ with a fractional value of 0.21 would be identified as falling outside the green list interval in $P_2$. Consequently, a new value would be randomly chosen from the nearest green list interval, 
$\color{teal}[0.3,0.4]$. For instance, 
$0.31$ could be selected as the watermarked value for 
$0.21$. This procedure is repeated for each subsequent data point to generate a fully watermarked tabular dataset. Given $m$, we can use $(x-i(x)) / (1 / 2m)$ to find the nearest ``green list'' intervals pair and consequently the closest interval with $O(1)$ time complexity. Therefore, the overal time complexity for watermarking an $n \times p$ tabular dataset is $O(np)$. We provide python-style pseudocode in \cref{appx:pseudo_code} to facilitate understanding of the watermarking scheme.

\paragraph{Tabular watermark with marginal data distortion}
We establish the following theorem concerning the impact on data fidelity of our watermarking approach.

\begin{theorem}[Fidelity]
\label{thm:fidelity}
Let $\mathbf{X}$ be a $n \times p$ dataframe, and let $\mathbf{X}_w$ denote its watermarked version. Conditioned on $\mathbf{X}$, it holds with probability one that
$$\|\mathbf{X}_w -\mathbf{X}\|_{\infty} \leq \frac{1}{m},$$
where $m$ is the number of ``green list'' intervals, a parameter controlling the granularity of the watermarking process. 
\end{theorem}
We refer to \cref{appx:proof_thm_fidelity} for the full proof of \cref{thm:fidelity}. A corollary naturally emerges that establishes an upper bound on the Wasserstein distance between $\mathbf{X}_w$ and $\mathbf{X}$ based on \cref{thm:fidelity}, providing a quantifiable measure of the distance between the two distributions.
\begin{corollary}
\label{corollary:wasserstein}
Let $F_{\mathbf{X}}=\sum_{j=1}^{n}\frac{1}{n}\delta_{\mathbf{X}[j,:]}$ be the empirical distribution built on $\mathbf{X}$, let $F_{\mathbf{X}_w}=\sum_{j=1}^{n}\frac{1}{n}\delta_{\mathbf{X}_w[j,:]}$ built on $\mathbf{X}_w$, then it holds with probability one that
\begin{align}
\mathcal{W}_k\left(F_{\mathbf{X}}, F_{\mathbf{X}_w}\right)\leq \frac{p^{\frac{1}{2}}}{m}, 
\end{align}
where $\mathcal{W}_k$ is the $k$-Wasserstein distance. 
\end{corollary}

\begin{remark}
\cref{thm:fidelity} assures us that the proposed watermark has marginal impact on the data fidelity. Specifically, it indicates that by increasing $m$ to sufficiently refine the granularity of the intervals (i.e., the length of each interval is $1/(2m)$), the watermarked data $\mathbf{X}_w$ will closely approximate the original $\mathbf{X}$, with an error rate of $1/m$ (the bounds in \cref{thm:fidelity,corollary:wasserstein} are tight). This property is crucial for ensuring that the fidelity of the data is maintained, while still embedding a robust watermark, as we will see later.
\end{remark}

\section{Detection of the Tabular Data Watermark}
\label{sec:detection}
Similar to \cite{kirchenbauer2023watermark}, the detection of watermarks in tabular data is conceptualized within a theoretical framework that transforms the process into a hypothesis-testing problem. In this context, we first introduce a theorem that solidifies the theoretical underpinnings of watermark detection:
\begin{lemma}[Prelim. for detection]
\label{thm:preliminary}
Consider a probability distribution \(F\) with a continuous probability density function \(f\). 
As \(m \rightarrow \infty\), 
$$P_{x\sim F}(x-i(x) \in G) \to \frac{1}{2},$$ 
where $i(x)$ is the integer part of $x$, such that $x \in [i(x),i(x)+1)$; $G$ represents the set of green list intervals. $x-i(x)$ therefore specifies the fractional part of the data point $x$. 
\end{lemma}
\begin{remark}
From the proof of \cref{thm:preliminary} (see \cref{appx:proof_thm_prelim}), we can see that this convergence is in fact consistent on all the possible choices of the green list.
\end{remark}

We formulate the task of detecting watermarks as a hypothesis-testing problem:

\begin{center}
\(H_0\): The table is not watermarked. vs. \(H_1\): The table is watermarked.
\end{center}

Based on the theoretical result in \cref{thm:preliminary}, we can claim that for any column of continuous variables in the given tabular data, the probability of an element falling into the ``green list'' intervals approximates $\frac{1}{2}$, when $m$ is large enough. This result is tangential to how the ``green list'' intervals are exactly chosen, \ie, any possible choice of these intervals following the procedure in \cref{alg:wm_algo} would suffice. Let $\mathbf{T}_i$ denotes the number of elements in the $i$-th column that fall into the ``green list'' intervals. We can see that $\mathbf{T}_i$ approximately follows a binomial distribution $B(n, \frac{1}{2})$ under $H_0$ when $m$ is large. For a particular value $t_i$ of $\mathbf{T}_i$, the p-value can be calculated using $\mathbf{P}(B(n, \frac{1}{2}) \geq t_i)$ to determine how statistically significant $t_i$ is.

When $n$ is large, by the central limit theorem, we can further model $T_i$ by 
\begin{align}
2\sqrt{n}(\frac{T_i}{n}-\frac{1}{2}) \to N(0,1).
\end{align}

To extend the analysis to tabular data with multiple columns, we need to consider the joint distribution across all all \(p\) columns.
We present the following theorem that indicates a quite surprising result which we term as the ``asymptotic independence'' of the watermarked column distributions. Specifically, for a random sample (one random row) of the $n\times p$ table, \ie, $\mathbf{x}=(x_1,x_2,\cdots,x_p)$ generated from a distribution $F$ with continuous probability density function, the events $\{x_i-i(x_i)\in G\},i=1,2,\cdots,p$ are independent when $m\to \infty$:
\begin{theorem}
[Asymptotic independence]
\label{thm:preliminary_p_dimensional}
Consider a $p$-dimensional probability distribution $F$ with continuous probability density function $p(x_1,x_2,x_3,\cdots,x_p)$, then as $m \to \infty$,
\begin{align}
\begin{split}
\mathbf{P}_{x \sim F}(\bigcap_{i=1}^{p} A_i )\to (\frac{1}{2})^{p},
\end{split}
\end{align}
where $A_i \in \{ \{x-i(x)\in G \}, \{ x-i(x)\notin G \}   \}$
\end{theorem}

\begin{remark} \cref{thm:preliminary_p_dimensional} implies that when $m$ is large enough, $\{T_i\}_{i=1}^p$ are independent random variables. Of note, the independence shown above \textit{does not} require independence of the data distribution, making this statement especially non-trivial; it holds for any continuous density functions. We can see from the proof in \cref{appx:proof_p_dim} that the independence originates from our design of the watermarking process, and is induced by sufficiently large $m$.  
\end{remark}

Consequently, we can establish that as \(n\) approaches infinity, the sum of squared standardized deviations of \(T_j\) converges to a chi-squared distribution by definition:
\begin{align}
\sum_{j=1}^{p} \left[2\sqrt{n} \left(\frac{T_j}{n} - \frac{1}{2}\right)\right]^2 \to \chi^2_p.
\end{align}
Of note, in practical scenarios, the specific entries that are watermarked within the table could remain unknown. Consequently, it is imperative to consider all columns uniformly, which results in the chi-squared testing as detailed above.
Another practical concern is that sometimes it is possible to encounter datasets with high dimensionality, \ie, large $p$. We provide the following asymptotic result indicating that the $\chi_p^2$ statistics remain valid even when the dimension $p$ goes to infinity, as long as $p$ and $n$ goes to infinity with certain rates:
\begin{theorem}
\label{thm:high_dimensional}
Assume that $\{ T_i\}_{i=1}^p$ i.i.d. follows $B(n,\frac{1}{2})$, then as $n\to \infty$, if $p=o(n^\frac{2}{7})$, we have
\begin{align}
\sum_{j=1}^{p} \left[2\sqrt{n} \left(\frac{T_j}{n} - \frac{1}{2}\right)\right]^2 \xrightarrow{d} \chi^2_p.
\end{align}
still holds even if $p \to \infty$.
\end{theorem}
\color{black}

\section{Robustness of the Tabular Data Watermark}
In this section, we further examine the robustness of the proposed watermark when exposed to attacks. We assume that the attacker has no knowledge about the ``green list'' intervals. Since our detection framework is based on hypothesis-testing, the ultimate goal of this attack can be regarded as increasing p-value as much as possible. 

Specifically, we consider a scenario where the attacker alters \( k_i \) elements from green-listed to non-green-listed in the \(i\)-th column of a fully watermarked tabular dataset. The resultant chi-square statistic for this modification is 
$\sum_{j=1}^p 4n \left(\frac{1}{2} - \frac{k_i}{n}\right)^2$. 
Therefore, to assess the robustness of the proposed watermark, we identify the minimum of \(\sum_{j=1}^p k_i\) such that:
\begin{equation}
\label{equ:chi_stat_change}
\sum_{j=1}^p 4n \left(\frac{1}{2} - \frac{k_i}{n}\right)^2 \leq \chi^2_p(1-\alpha),
\end{equation}
where \(\chi^2_p(1-\alpha)\) is the \((1-\alpha) \%\) quantile of the chi-square distribution with \(p\) degrees of freedom. Typically, \(1-\alpha\) is set as \(0.95\). \cref{equ:chi_stat_change} quantifies the minimum number of elements required for a successful attack, \ie, increasing the p-value so that it is higher than $\alpha$.
By Mean Squared-Arithmetic Inequality, to achieve this target, we require:
\begin{align}
\frac{2 \sqrt{n} \sum_{j=1}^{p} \left(\frac{1}{2} - \frac{k_i}{n}\right)}{p} \leq  \sqrt{\frac{\sum_{j=1}^{p} 4n\left(\frac{1}{2}-\frac{k_i}{n}\right)^2}{p}} \leq \sqrt{\frac{\chi^2_p(1-\alpha)}{p}},
\end{align}
which consequently implies that:
\begin{equation}
\label{equ:minimum_k}
\sum_{j=1}^{p} k_j \geq \frac{1}{2}(np) - \frac{1}{2}\sqrt{np}\sqrt{\chi^2_p(1-\alpha)}.    
\end{equation}

\cref{equ:minimum_k} demonstrates a lower bound of the number of elements to be moved out of ``green list'' intervals for a successful attack. 

To further assess the robustness of our watermarking framework, we consider the common choice of attacking with additive noise, since a wide range of attacks can be generally regarded as adding noise to the watermarked data; different noise distributions correspond with different attacking strategies. We start with examining how these attacks influence the distribution of green-listed elements. The probability that this additive noise successfully moves an element out of the ``green list'' intervals (a.k.a. attack success rate) is specified by the following theorem:
\begin{theorem}[Attack success rate]
\label{attack_success_rate}
Given noise \(\epsilon\) following a (not necessarily zero mean) distribution \(A\) with a continuous probability density function, for any \(x_w\) whose fractional part lies within a green list interval, as \(m \to \infty\),
\[
P_{\epsilon \sim A}(x_w+\epsilon-i(x_w+\epsilon) \notin \textit{G}) \to \frac{1}{2},
\]
where \(i(\cdot)\) is the integer part, \ie, \(x_w+\epsilon \in [i(x_w+\epsilon), i(x_w+\epsilon)+1)\).
\end{theorem}
\begin{proof}
The idea for proving \cref{attack_success_rate} is as follows. Given $x_w$, $x_w+\epsilon$ is a random variable with a continuous density. In \cref{thm:preliminary} we have proved that the probability of any random variable with a continuous density falling into the ``green list'' intervals converges to $\frac{1}{2}$. Symmetrically, the probability of any random variable with a continuous density not falling into the ``green list'' intervals also converges to $\frac{1}{2}$. Applying this to $x_w+\epsilon$, we finish the proof of \cref{attack_success_rate}.
\end{proof}

Motivated by the preceding discussions, we would like to address a crucial scenario: if each element within the ``green list'' intervals has an upper-bounded attack success probability of \( q \leq \frac{1}{2} \) (note that $q \to \frac{1}{2}$ regardless of the distribution of $\epsilon$ as $m \to \infty$), how many elements must be attacked to ensure a p-value higher than $\alpha$, thereby indicating a successful attack? 

We therefore formalize the following theorem:

\begin{theorem}[Robustness]
\label{thm:robustness}
Consider a \( n \times p \) table \(\mathbf{X}\) with all elements initially in the ``green list'' intervals. If the prob. of successfully attacking each element is no more than \( \frac{1}{2} \), and \( \widehat{k}_i \) denotes the number of elements attacked in the \(i\)-th column, then an attack will fail---meaning the calculated p-value will be $\alpha$ or higher---with at least probability \( 1-e^{-\frac{1}{2}(\sqrt{np}-\sqrt{\mathcal{X}_p^2(1-\alpha)})} \) if:
\[
\sum_{j=1}^p \widehat{k}_j \leq \frac{1}{1+\frac{1}{(np)^{\frac{1}{4}}}}(np-\sqrt{np}\sqrt{\mathcal{X}_p^2(1-\alpha)}).
\]
\end{theorem}

\begin{remark}
\cref{thm:robustness} underscores that if the success probability of attacking an individual element is capped at \( \frac{1}{2} \), then even attacking \((1+o(1))np\) elements is insufficient to significantly increase the likelihood of overcoming the hypothesis test. This result implies that an extensive number of targeted attacks is required to disrupt the hypothesis-testing mechanism effectively.
\end{remark}

\section{Experiments}
\label{sec:exps}
We now empirically evaluate our proposed tabular data watermark regarding \textbf{i) fideliy}, \textbf{ii) detection rate} and \textbf{iii) robustness} on synthetic and real-world datasets. We kindly refer to \cref{appx:exp_set} for additional experiment settings, results and discussions.

\subsection{Synthetic Dataset Examples} 
\label{sec:toy_data}
We first evaluate our method with synthetic data as a quick sanity check to validate our theoretical results. As the proof-of-concept experiments, we use Gaussian data to show that our framework can indeed greatly maintain data fidelity, demonstrate satisfying detection rates and achieve appealing robustness against attacks with additive noise. We set $m=1000$ for the following experiments.

\paragraph{Fidelity}
We start with evaluating the impact of our tabular data watermark on data fidelity with single-column data. Specifically, we draw a $2000 \times 1$ table from standard Gaussian to embed our proposed watermark. We can see from the kernel density estimation results in \cref{subfig:kde_syn_before,subfig:kde_syn_after} that our proposed watermark has negligible impact on the original data distribution, consistent with our statement in \cref{thm:fidelity,corollary:wasserstein}. We provide quantitative results on real datasets in \cref{tab:wasserstein_dist} and \cref{appx:add_exp}. We further evaluate the impact of our watermark on correlated multi-column data. To enforce the correlation between columns, we iteratively generate each column as 
$X_{j + 1} = 1.1 X_{j} + \epsilon$ if $j$ is odd, or $X_{j + 1} = X_{j} / 1.1 + \epsilon$ if $j$ is even, $j=1,...,p$. $X_{j}$ denotes the $j$-th column data, and $\epsilon \sim N(0, I_n)$. In this experiment, we construct a $10000 \times 1000$ table, and calculate the correlation matrices before and after applying our tabular data watermark ($m=1000$) to probe the impact. We can see from \cref{subfig:corr_no_wm,subfig:corr_wm_all,subfig:corr_diff} that the proposed watermark demonstrates marginal influences on the statistical relation among columns, with a maximum absolute difference of correlation values of $\sim 0.01$.

\begin{wrapfigure}{r}{0.53\textwidth}
    \vspace{-20pt}
    \centering
    \subfigure[KDE before wm.]{
        \includegraphics[width=0.45\linewidth]{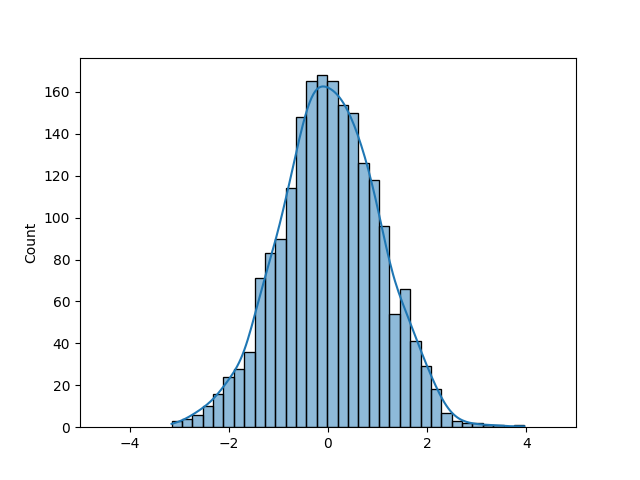}
        \label{subfig:kde_syn_before}
    }
    \subfigure[KDE after wm.]{
        \includegraphics[width=0.45\linewidth]{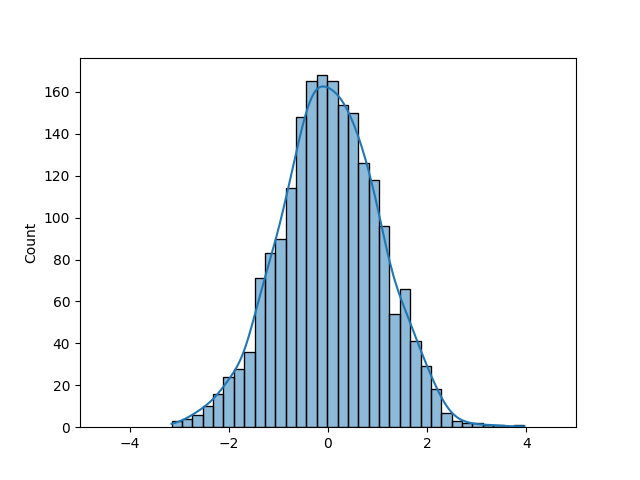}
        \label{subfig:kde_syn_after}
    } 
    \caption{\textbf{KDE plots for the Gaussian data w/ and w/o our proposed watermark;} \texttt{wm} as the shorthand of our watermark; figs and tabs henceforth follows this format.}
    \centering
    \subfigure[corr. w/o wm.]{
        \includegraphics[width=0.29\linewidth]{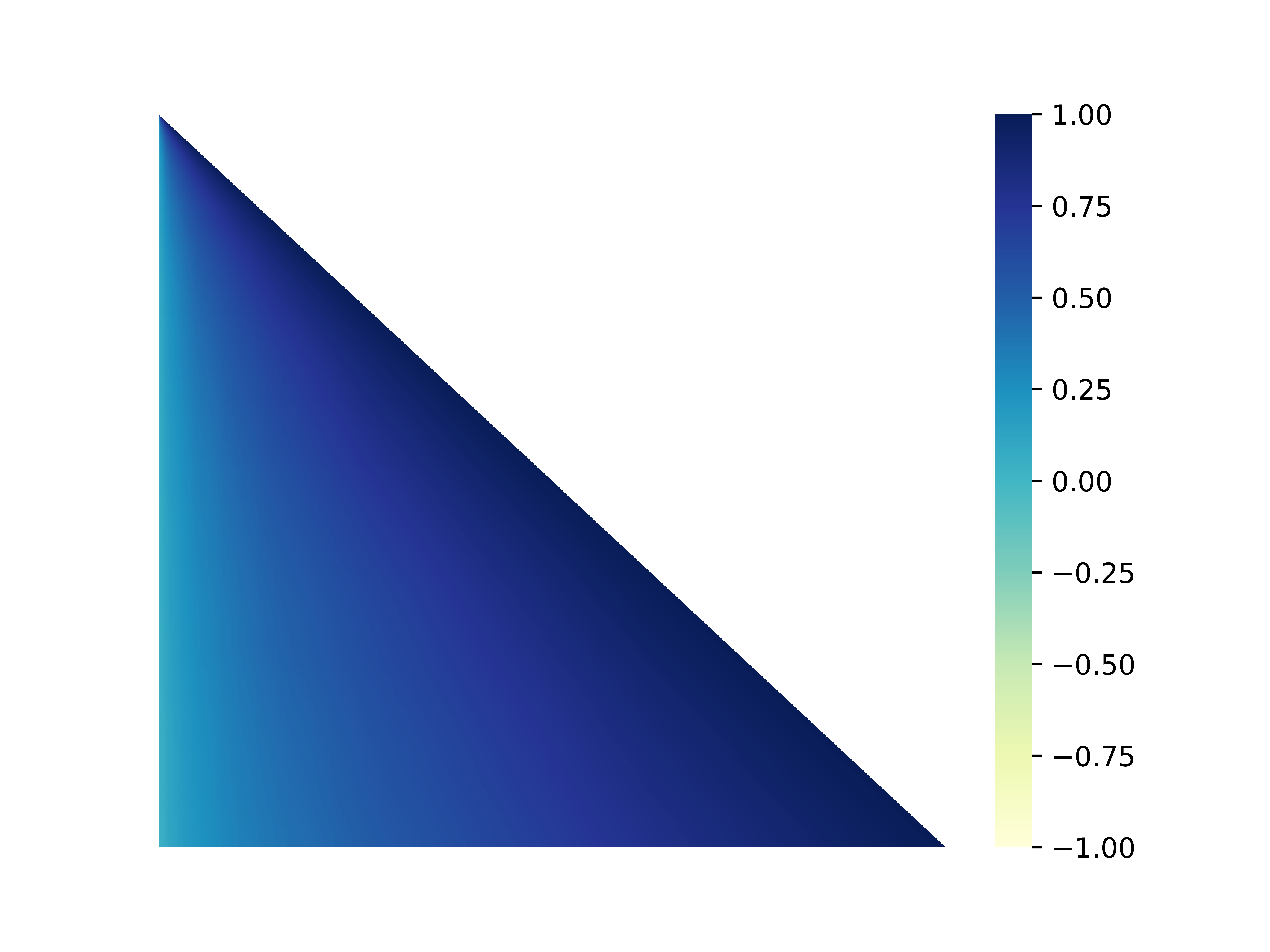}
        \label{subfig:corr_no_wm}
    }
    \subfigure[corr. w/ wm.]{
        \includegraphics[width=0.29\linewidth]{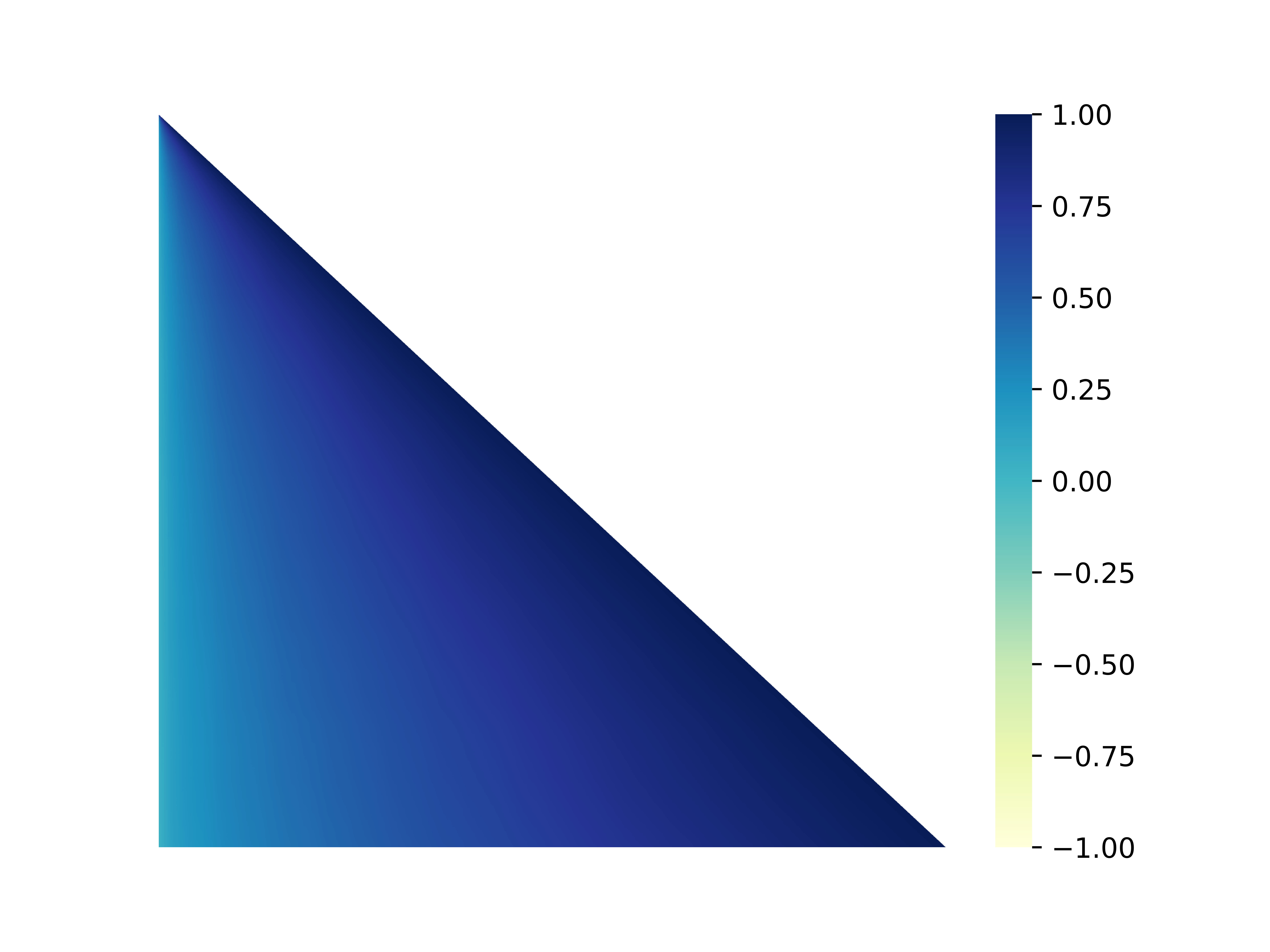}
        \label{subfig:corr_wm_all}
    }
    \subfigure[Diff. of corr.]{
        \includegraphics[width=0.29\linewidth]{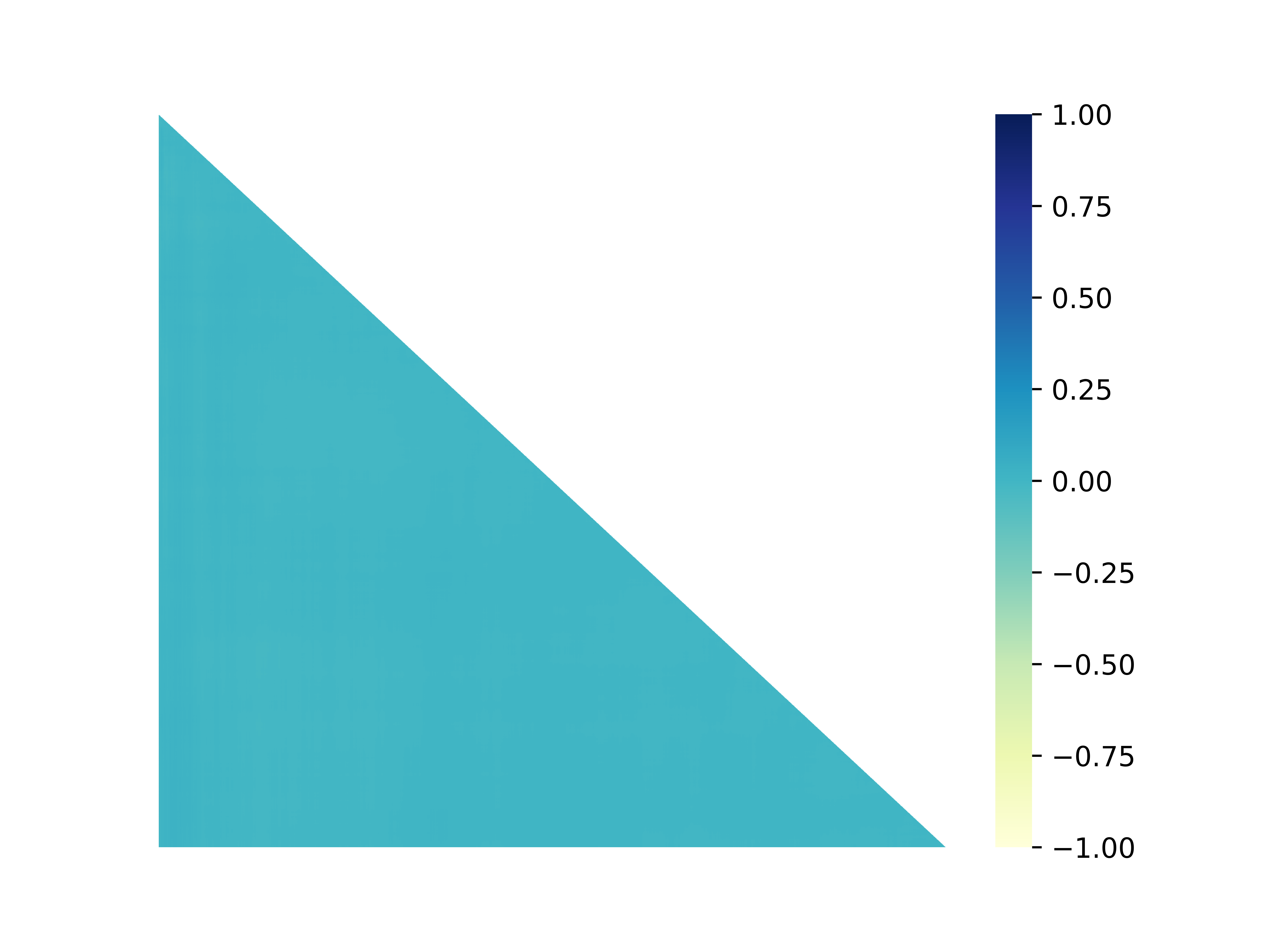}
        \label{subfig:corr_diff}
    }
    \caption{\textbf{Visualization of correlation matrices} and the difference between the correlation matrices w/ and w/o applying our proposed watermark. Zoom-in for more details.}
    \label{fig:viz_kde_corr_mat}
    \vspace{-10pt}
\end{wrapfigure}

\paragraph{Detection rate (True Postive Rate)} For multi-column tabular data, we consider two scenarios including i) adding the watermark to only one column, as a stress test to examine the effectiveness of our approach with extremely limited computational resources, and ii) adding the watermark to all columns in the table, which is the standard case. We evaluate the detection rates when applying our watermark to tables with different number of rows and columns (see \cref{subfig:chi_sing_col,subfig:chi_mult_col}). The true negative rate is 1 in all settings. We refer to \cref{appx:syn_dat,appx:add_exp} for details and ROC-AUC scores. In \cref{subfig:chi_sing_col}, we observe that the watermark is still largely detectable even when only one column is watermarked. We can see that the detection rate under this particular circumstance is high as long as the number of rows is sufficiently large. In \cref{subfig:chi_mult_col}, the detection rate is constantly high regardless of the size of the table, confirming the effectiveness of our approach. We refer to \cref{appx:add_exp} for additional results of simulating high-dimensional tables, where the column number $p$ exceeds the row number, \eg $p = 100n$. Our watermark can still be effectively detected with near perfect rates.

\begin{figure}[!h]
    \vspace{-7pt}
    \centering
    \subfigure[wm. one col.]{
        \includegraphics[width=0.22\linewidth]{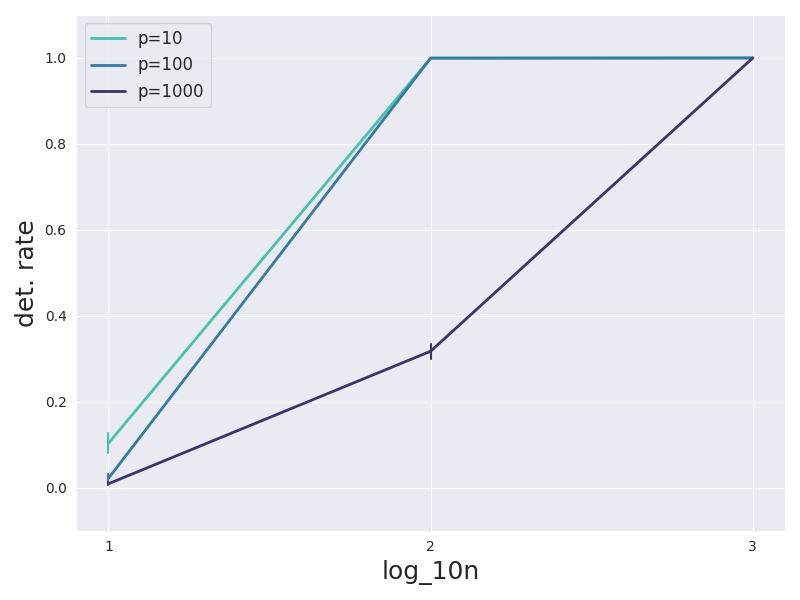}
        \label{subfig:chi_sing_col}
    }
    \subfigure[wm. all cols.]{
        \includegraphics[width=0.22\linewidth]{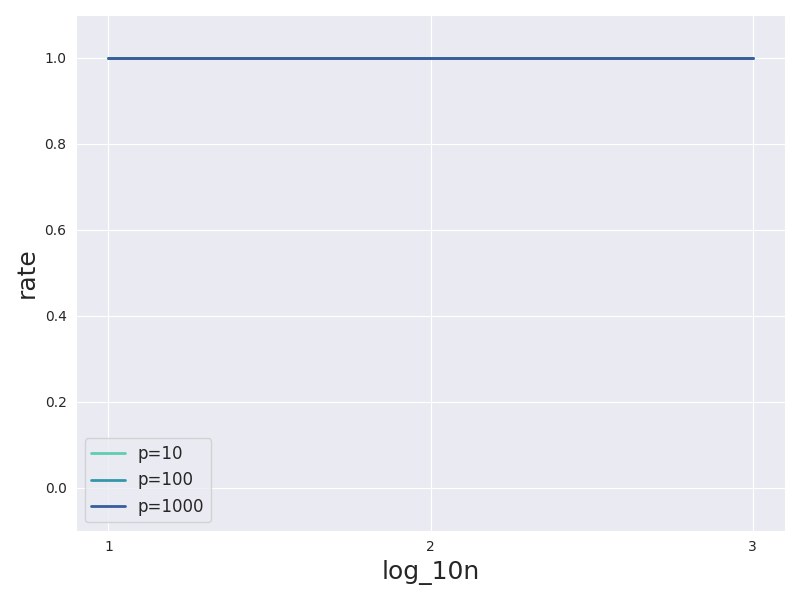}
        \label{subfig:chi_mult_col}
    }
    \subfigure[wm. one col. (atck)]{
        \includegraphics[width=0.22\linewidth]{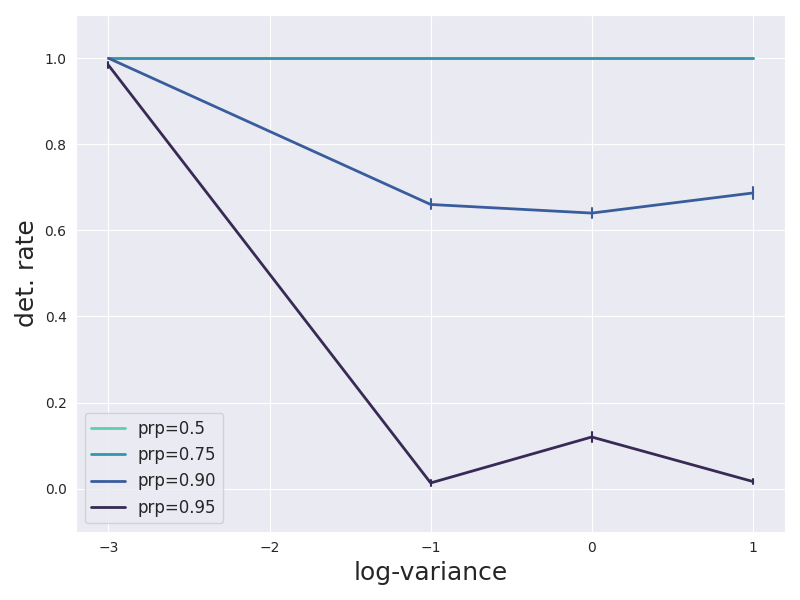}
        \label{subfig:chi_sing_col_atck}
    }
    \subfigure[wm. all cols. (atck)]{
        \includegraphics[width=0.22\linewidth]{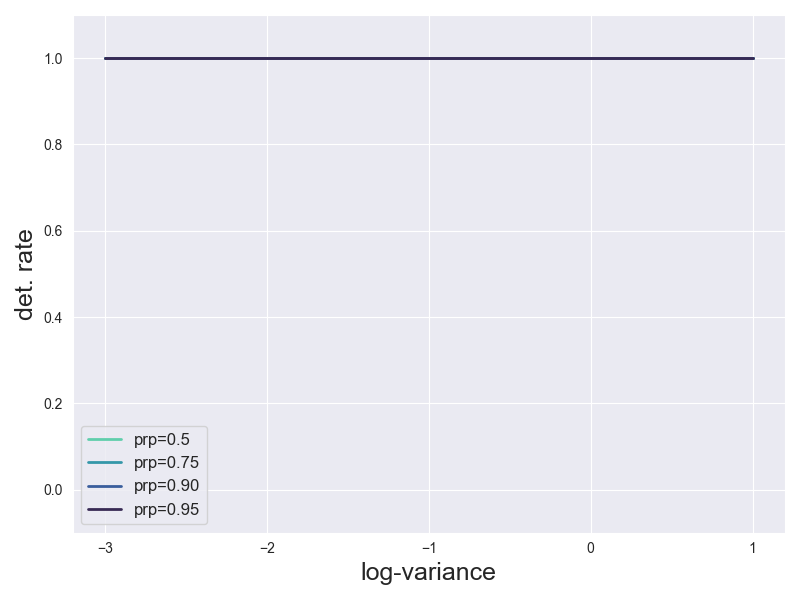}
        \label{subfig:chi_mult_col_atck}
    }
    \caption{\textbf{Detection rates of the proposed watermark applied to tabular data with different number of rows and columns.} In \cref{subfig:chi_sing_col_atck,subfig:chi_mult_col_atck} we plot the detection rates of the watermark after adding noises with different level of variances (in $\log_{10}$ scale). \texttt{prp} is the proportion of the elements in a table being modified. Rates over $1000$ independent samples; error bars over 3 runs. Zoom-in for more details.}
    \label{fig:plot_det_rate}
    \vspace{-7pt}
\end{figure}

\paragraph{Robustness} 
In these experiments, we apply the watermark to tables of size $5000 \times 100$ as the representative, and then add Gaussian noises with different variances to perturb different proportion of the watermarked data. In \cref{subfig:chi_sing_col_atck}, it is demonstrated that watermarking just a single column (approx. 1\% of the original data) already has decent robustness when 75\%
of the elements are modified by the attacker. In \cref{subfig:chi_mult_col_atck}, the detection rate is constantly high regardless of the variance of the added noise or the proportion (can be as high as $\sim95\%$) of the elements being modified, as long as all columns are watermarked. The result in \cref{subfig:chi_mult_col_atck} is particularly encouraging as the variance of the added noise can be as large as 10, while the variance of data distribution is only approximately 1. The results support our theoretical analysis in \cref{thm:robustness}.

\subsection{Results on Generative Tabular Data}
In this section, we extensively evaluate our watermarking framework on real-world datasets to check the effectiveness of our approach.

\paragraph{Datasets \& tab. generators} Specifically, we employ TVAE \cite{xu2019modeling}, CTABGAN \cite{zhao2021ctab} and TabDDPM \cite{kotelnikov2023tabddpm} as representatives of VAE-based \cite{kingma2013auto}, GAN-based \cite{goodfellow2014generative}, and DDPM-based \cite{ho2020denoising} tabular data generators to generate tabular data. For systematic investigation of our tabular data watermark performance on these tabular generative models, we consider a diverse set of 6 real-world public datasets with various sizes, nature, number of features, and their distributions; these datasets are commonly used for tabular model evaluation \cite{zhao2021ctab,gorishniy2021revisiting}. We provide additional details of datasets, evaluation measure, and tunning process for tabular data generators in \cref{appx:real_dat}.

\paragraph{Practical implementation} In practice, some columns in certain generated datasets can follow ill-shaped distributions, \eg, some distributions have spikes concentrated on certain values, which may violate the assumption of our framework (see \cref{appx:add_exp}). To be specific, from \cref{thm:preliminary} we know that as $m$ tends to infinity, the probability of an element falling within a ``green list'' interval converges to 1/2. 
The rate at which this convergence occurs, however, depends largely on the smoothness of the distribution. To address this issue, we therefore adopt a heuristic approach by selecting columns with relatively smooth data distributions to embed the watermark. 

Our heuristic approach assumes that for a column data distribution with enough smoothness, the probability of an element falling within a ``green list'' interval, denoted as $\widehat{p}$, lies with in $[\frac{1}{2} - \Delta, \frac{1}{2} + \Delta]$, when $m$ is not sufficiently large. We can therefore use the frequency $\widehat{f}\approx\widehat{p}$ of an element falling within ``green list'' intervals as an indicator of the distribution smoothness. Specifically, to filter out the columns with low smoothness, we set $\Delta = 0.01$ and $m$ to be within the range $\{1000, 1500, 2000, 2500, ..., 4500, 5000\}$ and count for each column how many times $\widehat{f}$ falls outside the range $[\frac{1}{2} - \Delta, \frac{1}{2} + \Delta]$ with different choices of $m$. Specifically, we sweep $m$ over its range, and repeat the experiment 5 times for each value of $m$. If the number exceeds 10\% of the total number ($5 \times 9$) of experiments, then we identify this column as a non-smooth column and discard it (see \cref{appx:add_exp} for further results of this filtering process; most generated datasets remain untouched). For the rest of these columns, we choose for each column the $m$ that maximize the number of times that $\widehat{f}$ falls inside the range $[\frac{1}{2} - \Delta, \frac{1}{2} + \Delta]$ to conduct the following experiments. For each column, we normalize the distribution to zero mean and unit variance before adding our proposed watermark.

\paragraph{Fidelity} 
We verify the distribution distance between the original generated data and watermarked data is indeed of $O(\frac{1}{m})$. In \cref{tab:wasserstein_dist}, we calculate the Wasserstein-1 distance between the empirical distribution generated by TabDDPM (see \cref{appx:add_exp} for results of TVAE and CTABGAN) and its watermarked version; we provide the distance between real data and generated data for reference.
\begin{table}[!htbp]
    \vspace{-5pt}
    \centering
    \caption{\textbf{Wasserstein-1 distance between generated data distribution and watermarked data distribution.}}
    \vskip 0.1in
    \begin{tabular}{ccccccc}
    \toprule
        {} & California & Gesture & House 
           & Wilt & Higgs-small & Miniboone \\
    \midrule
    Orig2Gen & 0.0222 & 0.0602 & 0.0315
             & 0.0767 & 0.0142 & 0.0161 \\
    Gen2Watermarked 
             & 0.0004 & 0.0004 & 0.0004 
             & 0.0005 & 0.0003 & 0.0001 \\
    \midrule
    $m$      & 1000   & 1000   & 1000   
             & 1000   & 1000   & 2500   \\
    $1 / m$  & 0.001 & 0.001   & 0.001  
             & 0.001 & 0.001   & 0.0004 \\
    \bottomrule
    \end{tabular}
    \label{tab:wasserstein_dist}
    \vspace{-5pt}
\end{table}

\paragraph{Accuracy of the watermarked tabular data}  
We summarize the effectiveness (measured by ROC-AUC scores) of our tabular data watermark on the generated tabular data. We can see in \cref{tab:det_real} that our method demonstrates desirable accuracies on these generated datasets.
\begin{table}[!htbp]
    \vspace{-5pt}
    \centering
    \caption{\textbf{Accuracy (ROC-AUC score) of the tabular data watermark.} }
    \vskip 0.1in
    \begin{tabular}{ccccccc}
    \toprule
        {} & California & Gesture & House 
           & Wilt & Higgs-small & Miniboone \\
    \midrule
    TVAE & 1.000 & 1.000 & 1.000 & 1.000 & 1.000 & 1.000 \\
    CTABGAN & 1.000 & 1.000 & 1.000 & 1.000 & 1.000 & 1.000 \\
    TabDDPM & 1.000 & 1.000 & 1.000 & 1.000 & 1.000 & 0.999 \\
    \bottomrule
    \end{tabular}
    \label{tab:det_real}
    \vspace{-5pt}
\end{table}

\paragraph{Robustness} To examine the robustness of our approach, we add zero mean Gaussian noise with the standard variance as $0.01\widehat{\sigma}$ to perturb the watermarked tabular data; $\widehat{\sigma}$ represents the standard variance of the watermarked tabular data. We choose this relatively small noise variance to make sure that the ML efficiency (or utility) \cite{xu2019modeling} of the generated data is not severely deteriorated by the attacks, while the added noise can distort the watermark as much as possible. This is consistent with most practical scenarios, where the attacker intends to remove the watermark as much as possible while preserving the original data information (\eg, \cite{kirchenbauer2023watermark,zhao2023invisible}). The noise are added to 95\% of all the elements in the watermarked tabular data. We can see from \cref{tab:det_robust} that our watermark can still be reliably detected on all the generated datasets. 
\begin{table}[!htbp]
    \vspace{-5pt}
    \centering
    \caption{\textbf{Accuracy (ROC-AUC score) of the tabular data watermark after additive noise attack.}}
    \vskip 0.1in
    \begin{tabular}{ccccccc}
    \toprule
        {} & California & Gesture & House 
           & Wilt & Higgs-small & Miniboone \\
    \midrule
    TVAE & 1.000 & 1.000 & 1.000 & 1.000 & 1.000 & 1.000 \\
    CTABGAN & 1.000 & 1.000 & 1.000 & 1.000 & 1.000 & 1.000 \\
    TabDDPM & 1.000 & 1.000 & 1.000 & 1.000 & 1.000 & 0.999 \\
    \bottomrule
    \end{tabular}
    \label{tab:det_robust}
    \vspace{-10pt}
\end{table}

\paragraph{Utility} To examine the impact of our watermark on the utility of the generated data, we follow the evaluation protocol in \cite{kotelnikov2023tabddpm} and train CatBoost classifiers \cite{prokhorenkova2018catboost} using the watermarked generated data and the original generated data and compare the performances. We can see in \cref{tab:utility} that our tabular data watermark has negligible impact on the utility of the synthesized data. 

\begin{table}[!htbp]
    \centering
    \caption{\textbf{Impact of the tabular data watermark on utility.} The metrics used for each dataset is provided after the dataset name. Results calculated using 20 generated copies of the original dataset.}
    \vskip 0.1in
    \begin{tabular}{ccccc}
    \toprule
    {} & {Generator} & California (R2) 
       & Gesture (F1) 
       & House (R2) \\
    \midrule
    {} & TVAE & $0.736 \pm 0.004$ 
         & $0.418 \pm 0.012$ 
         & $0.448 \pm 0.010$ \\
    orig. data & CTABGAN & $0.577 \pm 0.007$ 
            & $0.411 \pm 0.005$ 
            & $0.327 \pm 0.008$ \\
    {} & TabDDPM & $0.823 \pm 0.003$ 
                 & $0.575 \pm 0.009$ 
                 & $0.638 \pm 0.007$ \\
    \midrule
    \rowcolor{gray}
    {} & TVAE & $0.735 \pm 0.037$ 
         & $0.419 \pm 0.012$ 
         & $0.449 \pm 0.011$ \\
    \rowcolor{gray}
    wm. data & CTABGAN & $0.577 \pm 0.007$ 
            & $0.409 \pm 0.009$ 
            & $0.328 \pm 0.007$ \\
    \rowcolor{gray}
    {} & TabDDPM & $0.823 \pm 0.003$ 
                 & $0.554 \pm 0.007$ 
                 & $0.638 \pm 0.007$ \\
    \midrule
    {} & {Generator} & Wilt (F1) 
       & Higgs-small (F1) 
       & Miniboone (F1) \\
    \midrule
    {} & TVAE & $0.500 \pm 0.020$ 
              & $0.665 \pm 0.001$
              & $0.905 \pm 0.002$ \\
    orig. data & CTABGAN & $0.666 \pm 0.019$ 
                         & $0.602 \pm 0.004$ 
                         & $0.852 \pm 0.002$ \\
    {} & TabDDPM 
              & $0.892 \pm 0.017$ 
              & $0.713 \pm 0.002$ 
              & $0.931 \pm 0.001$ \\
    \midrule
    \rowcolor{gray}
    {} & TVAE & $0.494 \pm 0.020$ 
              & $0.665 \pm 0.001$ 
              & $0.905 \pm 0.002$ \\
    \rowcolor{gray}
    wm. data & CTABGAN 
             & $0.666 \pm 0.019$ 
             & $0.601 \pm 0.005$ 
             & $0.852 \pm 0.002$ \\
    \rowcolor{gray}
    {} & TabDDPM 
        & $0.886 \pm 0.013$ 
        & $0.713 \pm 0.002$ 
        & $0.930 \pm 0.001$ \\
    \bottomrule
    \end{tabular}
    \label{tab:utility}
\end{table}

\section{Conclusion}
This paper presents a new watermarking method for tabular data to ensure the fidelity of synthetic datasets. The approach embeds watermarks into finely segmented data intervals, using a "green list" technique to minimize distortion and retain high data fidelity. A robust statistical hypothesis-testing framework is then proposed allowing for reliable detection of the watermarks, even in the presence of additive noise with large variances.
Experimental results demonstrate the effectiveness of the technique, with near-perfect detection rates in terms of AUC. The watermarking process shows high robustness against Gaussian noise attacks while having minimal impact on data utility, indicating its usefulness in practical scenarios where ML efficency for downstream tasks is of primary concern.
This work contributes to enhancing the security of both synthetic and real-world datasets, which is critical in the context of AI and machine learning applications. Future research could focus on improving the robustness of the watermarking method and extending its applicability across different types of data, for example, the categorical data.

\newpage

\bibliographystyle{unsrtnat}
\bibliography{ref}


\clearpage

\renewcommand\thefigure{A\arabic{figure}}
\setcounter{figure}{0}
\renewcommand\thetable{A\arabic{table}}
\setcounter{table}{0}
\renewcommand\theequation{A\arabic{equation}}
\setcounter{equation}{0}
\pagenumbering{arabic}
\renewcommand*{\thepage}{A\arabic{page}}

\newpage
\appendix
\section{Related Works}
\label{appx:related_works}
With the burgeoning success of generative models, there has been an increasing focus on integrating watermarking techniques into these models to enhance security and traceability recently. 

\paragraph{LLM watermark} One class of watermarking techniques for text generated by LLMs is based on dividing the vocabulary into ``green lists'' and ``red lists''. This line of works shares the spirit of our proposed tabular data watermark, but differs significantly from the methodological perspective. Representative works of this interesting direction include \cite{kirchenbauer2023watermark} and \cite{zhao2023provable}. \cite{kirchenbauer2023watermark} embeds watermarks by prioritized sampling of randomized ``green list'' tokens during text generation. It considers ``hard'' and ``soft'' embedding of ``green list'' tokens, which demonstrates great empirical effectiveness on the modern LLMs such as the OPT model. \cite{zhao2023provable} introduces a framework that infuses binary signatures into LLM-generated text using a learning-based approach. The framework features three key components: a message encoding module to embed the binary signatures, a reparameterization module to transform the encoded messages for reliable embedding, and a decoding module to extract the watermarks. An optimized beam search algorithm is employed to ensure the watermarked text remains coherent and consistent. 

Concurrently, watermarking LLM-generated text from the cryptographic perspective has led to fruitful and inspiring results in this direction. \cite{christ2023undetectable} introduces a watermarking scheme that embeds undetectable watermarks into generated text by modifying token selection probabilities using cryptographic techniques. This makes the watermark detectable only by those with a secret key. Similarly, \cite{zhao2024permute} presents the Permute-and-Flip (PF) decoder, which offers a watermarking scheme specifically tailored for the PF decoder. This scheme aims to maintain text quality and robustness while performing well in terms of perplexity and the detectability of watermarked texts. Additionally, there are methods that embed watermarks directly into the weights of LLMs \cite{zhao2022distillation,zhao2023protecting,xu2024learning}. Specifically, \cite{zhao2022distillation} introduces Distillation-Resistant Watermarking (DRW), which injects watermarks into prediction probabilities using a sinusoidal signal. The proposed watermark can be effectively detected by model probing, without inducing significant performance loss of the original model. \cite{zhao2023protecting} presents GINSEW, embedding invisible watermarks into probability vectors during text generation, which is detectable only with a secret key and robust against synonym randomization attacks. \cite{xu2024learning} introduces a reinforcement learning-based framework that co-trains a LLM and a detector to embed watermarks into model weights. The proposed framework has an emphasized robustness against adversarial attacks while successfully maintaining model utility.

\paragraph{Watermarking generated image data} Watermarking generative image data has drawn growing interest especially in recent years. In the pioneering work of \cite{wen2024tree}, the proposed framework embeds watermarks into the initial noise vector of diffusion models during the sampling process, resulting in surprisingly resilient, effective and invisible image data watermark. \cite{fernandez2023stable} finetunes the decoder of a diffusion model to embed watermarks directly into generated images, ensuring high detection accuracy and robustness against modifications. Both approaches emphasize the importance of embedding watermarks during the generation process to ensure invisibility and robustness.

\paragraph{Challenges in watermarking AI-generated data} Recent studies highlight significant theoretical and practical challenges in watermarking AI-generated content. \cite{zhang2023watermarks} proves the theoretical impossibility of simultaneously creating i) strong watermarks that cannot be removed by a computationally bounded attacker and ii) watermarks that does not significantly degrade data quality. As a counterexample, it constructs a random-walk-based attack that preserves content quality while effectively removing watermarks. \cite{zhao2023invisible} demonstrates the practical vulnerability of invisible watermarks, showing that regeneration attacks with noise addition and image reconstruction via generative models can remove up to $99\%$ of watermarks without significant quality loss. These findings emphasize the need to shift from invisible to semantically visible watermarks for robust protection.

\clearpage
\newpage
\section{Proof of Main Theorems}
\label{appx:proofs}

\subsection{Proof of \cref{thm:fidelity}}
\label{appx:proof_thm_fidelity}
\begin{proof}
$\forall x \in \mathbf{X}$, assume $x-i(x)$ lies in the $j th$ pair of consecutive intervals $[\frac{2j-2}{2m},\frac{2j-1}{2m}] \cup [\frac{2j-1}{2m},\frac{2j}{2m}]$. In our watermarking process, we resample a value $x_w$ from the nearest interval in the ``green list'' intervals $G$ to replace the $x$. We can see that 
\begin{align*}
\argmin _{Y \in G} \min_{y \in Y} d(x,y) =\argmin_{Y \in G} \max_{y \in Y} d(x,y).
\end{align*}
Assume $Y^*$ is the chosen nearest interval from the green list, and $Y^{**}$ is the interval chosen to be in the green list in the group $\{[\frac{2j-2}{2m},\frac{2j-1}{2m}], [\frac{2j-1}{2m},\frac{2j}{2m}]\}$,
we then have
\begin{align*}
d(x,x_w)\leq \max_{y \in Y^*} d(x,y) \leq \max_{y\in Y^{**} } d(x,y)\leq \max_{y\in [\frac{2j-2}{2m},\frac{2j-1}{2m}] \cup [\frac{2j-1}{2m},\frac{2j}{2m}]} d(x,y)\leq \frac{1}{m}
\end{align*}
\end{proof}

\subsection{Proof of \cref{corollary:wasserstein}}
\begin{proof}
The $k$-Wasserstein distance for two discrete measures 
\begin{equation}
\mu_0:=\sum_{i=1}^{k_0} a_{0 i} \delta_{x_{0 i}} \quad \text { and } \quad \mu_1:=\sum_{i=1}^{k_1} a_{1 i} \delta_{x_{1 i}},
\end{equation}
is defined as 
\begin{equation}
\label{definition_wasserstein}
\left[\mathcal{W}_k\left(\mu_0, \mu_1\right)\right]^k=\left\{\begin{aligned}
\min _{T \in \mathbb{R}^{k_0 \times k_1}} & \sum_{i j} T_{i j}\left|x_{0 i}-x_{1 j}\right|_2^k \\
\text { s.t. } & T \geq 0 \\
& \sum_j T_{i j}=a_{0 i} \\
& \sum_i T_{i j}=a_{1 j} .
\end{aligned}\right.
\end{equation}
For $F_{\mathbf{X}}$ and $F_{\mathbf{X}_w}$, if we take $T=diag 
\{\frac{1}{n},\frac{1}{n},\cdots,\frac{1}{n}\}$, $x_{0i}=\mathbf{X}[i,:],x_{1j}=\mathbf{X}_w[j,:]$ in \eqref{definition_wasserstein}, we could see
\begin{align}
\mathcal{W}_k\left(F_{\mathbf{X}},F_{\mathbf{X_w}} \right) \leq (\sum_{j=1}^{n}\frac{1}{n}\| \mathbf{X}[j,:]-\mathbf{X}_w [j,:]\|_2^k)^{\frac{1}{k}}\leq \frac{p^{\frac{1}{2}}}{m}
\end{align}
\end{proof}

\subsection{Proof of \cref{thm:preliminary}}
\label{appx:proof_thm_prelim}
\begin{proof}
We prove this by using the technique of truncation. $\forall \epsilon>0$, we could first choose $n$ large enough, so that 
\begin{align*}
\int_{-n}^{n} f(x)dx >1-\epsilon.
\end{align*}
Denote the ``green list'' intervals $G$ in $[0,1]$ as $\{ g_1(0),g_2(0),\cdots,g_m(0) \}$, where $g_i(0)$ is the interval chosen in the $i$-th group to be in the green list. We then define $g_k(j)$ as $g_k+j$, $\forall j \neq 0$. Therefore, 

\begin{align*}
\mathbf{P}(x-i(x) \in G )& =\mathbf{P}(x \in \bigcup_{j=-\infty}^{\infty} \bigcup_{k=1}^{m}g_k(j) )\\ & =\mathbf{P}(x \in \bigcup_{j=-\infty}^{-n-1} \bigcup_{k=1}^{m}g_k(j)  )+\mathbf{P}(x \in \bigcup_{j=n+1}^{\infty}\bigcup_{k=1}^{m} g_k(j) )+\mathbf{P}(x \in \bigcup_{j=-n}^{n}\bigcup_{k=1}^{m}g_k(j)),
\end{align*}

the first two terms of which could be bounded by 

\begin{equation*}
\begin{aligned}
\mathbf{P}(x \in \bigcup_{j=-\infty}^{-n-1} \bigcup_{k=1}^{m}g_k(j)  )
+
\mathbf{P}(x \in \bigcup_{j=n+1}^{\infty}\bigcup_{k=1}^{m} g_k(j) ) 
&\leq 
\int_{-\infty}^{-n}f(x)dx+\int_{n}^{\infty}f(x)dx \\
&=1-\int_{-n}^{n}f(x)dx <\epsilon.
\end{aligned}
\end{equation*}

We next consider the third term: we use $h_k(j)$ to denote the interval such that 
$g_k(j)\cup h_k(j)=[\frac{k-1}{m}+j,\frac{k}{m}+j]$. We can see that $h_k(0)$ is the complement of the $k$-th ``green list'' intervals in the $k$-th group $\{[\frac{2j-2}{2m},\frac{2j-1}{2m}],[\frac{2j-1}{2m},\frac{2j}{2m}] \}$, and $h_k(j)=h_k(0)+j$.

We then have 
\begin{equation*}
\begin{aligned}
| 
\mathbf{P}(x \in \bigcup_{j=-n}^{n}\bigcup_{k=1}^{m}g_k(j) ) -
\mathbf{P}(x \in \bigcup_{j=-n}^{n}\bigcup_{k=1}^{m}h_k(j) )|  
&\leq \sum_{j=-n}^{n}\sum_{k=1}^{m}| \mathbf{P}(x \in g_k(j))-\mathbf{P}(x \in h_k (j) ) | \\ 
&=     
\sum_{j=-n}^{n}\sum_{k=1}^{m}| \int_{x \in g_k(j)}f(x)dx-\int_{x \in h_k(j)} f(x)dx| \\
&\leq \sum_{j=-n}^{n}\sum_{k=1}^{m}\max_{x\in g_k(j),y\in h_k(j)}|f(x)-f(y) | \frac{1}{m} \\ 
&\leq \max_{k,j}\max_{x \in g_k(j),y \in h_k(j)} |f(x)-f(y) | (2n+1).
\end{aligned}
\end{equation*}

Using the result that a continuous function in a compact set is uniformly continuous, we can find some $m_0$ so that when $m\geq m_0$, namely $\frac{2}{m}\leq \frac{2}{m_0}$, 
\begin{align*}
\max_{k,j}\max_{x \in g_k(j),y \in h_k(j)} |f(x)-f(y)|\leq \max_{-n \leq x,y \leq n; |x-y|\leq \frac{2}{m_0} } |f(x)-f(y) | <\frac{\epsilon}{2n+1}.
\end{align*}

Under this circumstance, we have
\begin{align*}
|\mathbf{P}(x \in \bigcup_{j=-n}^{n}\bigcup_{k=1}^{m}g_k(j) ) -\mathbf{P}(x \in \bigcup_{j=-n}^{n}\bigcup_{k=1}^{m}h_k(j) )| \leq  \max_{k,j}\max_{x \in g_k(j),y \in h_k(j)} |f(x)-f(y) | (2n+1) <\epsilon,
\end{align*}
which implies
\begin{equation}
\label{eq:uniform_continous}
\begin{aligned}
2 
\mathbf{P}(x \in \bigcup_{j=-n}^{n}\bigcup_{k=1}^{m}g_k(j) )-\epsilon 
& \leq \mathbf{P}(x \in \bigcup_{j=-n}^{n}\bigcup_{k=1}^{m}g_k(j) )+\mathbf{P}(x \in \bigcup_{j=-n}^{n}\bigcup_{k=1}^{m}h_k(j) ) \\
& \leq 2 \mathbf{P}(x \in \bigcup_{j=-n}^{n}\bigcup_{k=1}^{m}g_k(j) )+\epsilon.
\end{aligned}
\end{equation}

Note that 
\begin{equation}
\label{eq:bound_for_sum}
1-\epsilon<\mathbf{P}(x \in \bigcup_{j=-n}^{n}\bigcup_{k=1}^{m}g_k(j) )+\mathbf{P}(x \in \bigcup_{j=-n}^{n}\bigcup_{k=1}^{m}h_k(j) )=\mathbf{P}(x \in [-n,n]) \leq 1,    
\end{equation}

Plugging \eqref{eq:uniform_continous} into \eqref{eq:bound_for_sum}, we then have
\begin{align*}
 1-2\epsilon < 2 \mathbf{P}(x \in \bigcup_{j=-n}^{n}\bigcup_{k=1}^{m}g_k(j) ) \leq 1+\epsilon,
\end{align*}

Therefore we have 
\begin{equation}
\label{eq:upper_bound_detect}
\begin{aligned}
{P({x}-i(x)\in G )} 
&=
\mathbf{P}(x \in \bigcup_{j=-\infty}^{-n-1} \bigcup_{k=1}^{m}g_k(j)  )+\mathbf{P}(x \in \bigcup_{j=n+1}^{\infty}\bigcup_{k=1}^{m} g_k(j) )+\mathbf{P}(x \in \bigcup_{j=-n}^{n}\bigcup_{k=1}^{m}g_k(j) ) \\
&\leq 
\epsilon+\frac{1+\epsilon}{2},
\end{aligned}
\end{equation}
and 
\begin{align}
\label{eq:lower_bound_detect}
P({x}-i(x)\in G )\geq \mathbf{P}(x \in \bigcup_{j=-n}^{n}\bigcup_{k=1}^{m}g_k(j) ) >\frac{1}{2}-\epsilon,
\end{align}
since $\forall \epsilon>0$, we have a $m_0$ so that when $m>m_0$, \eqref{eq:upper_bound_detect} and \eqref{eq:lower_bound_detect} hold, we finsh the proof.
\end{proof}

\subsection{Proof of \cref{thm:preliminary_p_dimensional}}
\label{appx:proof_p_dim}
To prove \cref{thm:preliminary_p_dimensional}, we first prove a lemma:
\begin{lemma}
\label{thm:preliminary_p_bounded_version}
Consider a $p$-dimensional probability distribution $F$ supported in $\|x\|_2 \leq R$ with continuous probability density function $p(x_1,x_2,x_3,\cdots,x_p)$, then as $m \to \infty$,
\begin{align}
\begin{split}
\mathbf{P}_{x \sim F}(\bigcap_{k=1}^{p} A_k )\to (\frac{1}{2})^{p},
\end{split}
\end{align}
where $A_i \in \{ \{x-i(x)\in G \}, \{ x-i(x)\notin G \}   \}$
\end{lemma}
\begin{proof}
Without loss of generality, we prove this result for $p=2$. The proof for $p>2$ is similar.
First, we prove that 
\begin{align}
\mathbf{P}(\bigcup_{j=1}^{\infty} \{p_{1}(x_1)\geq \frac{1}{j}  \})=1,
\end{align}
where $p_1$ is the marginal distribution of $x_1$. This is because 
\begin{align}
\mathbf{P}(\bigcup_{j=1}^{\infty}\{p_1(x_1)\geq \frac{1}{j} \} \bigcup \{p_1(x_1)=0 \})=1,
\end{align}
while $\mathbf{P}(p_1(x_1)=0)=0$.
Then by applying \cref{thm:preliminary} to $p_1(x_1)$, we have 
\begin{align}
\lim_{m \to \infty}\mathbf{P}(x_1-i(x_1)\in G \bigcap \bigcup_{j=1}^{\infty}\{p_1(x_1)\geq \frac{1}{j} \})=\lim_{m \to \infty }\mathbf{P}(x_1 -i(x_1) \in G)=\frac{1}{2}.
\end{align}

Since the sequence $\{ \ \bigcup_{j=1}^{N}\{p_1(x_1)\geq \frac{1}{j} \}, N=1,2,\cdots \}$ monotonically increases and converges to $\bigcup_{j=1}^{\infty} \{p_1(x_1)\geq \frac{1}{j} \}$,
we have
\begin{align}
\mathbf{P}(\bigcup_{j=1}^{N}\{p_1(x_1)\geq \frac{1}{j} \} ) \to \mathbf{P}(\bigcup_{j=1}^{\infty}\{p_1(x_1)\geq \frac{1}{j} \})=1, as \ N\to \infty.
\end{align}

Therefore, for any $\delta>0$,there exists $M_0$ and $N_0$, such that when $m>M_0$,
\begin{align}
\mathbf{P}(x_1-i(x_1) \in G \bigcap \bigcup_{j=1}^{N_0}\{  
p_1(x_1)\geq \frac{1}{j}
\})\geq \frac{1}{2}-\delta.
\end{align}

We further have
\begin{align}
\begin{split}
& \mathbf{P}(x_1-i(x_1) \in G \bigcap \bigcup_{j=1}^{N_0}\{  
p_1(x_1)\geq \frac{1}{j}
\},x_2-i(x_2)\in G)  \\= & \int_{x_1-i(x_1)\in G\bigcap \bigcup_{j=1}^{N_0} \{  
p_1(x_1)\geq \frac{1}{j}
\}} p_1(x_1)d x_1 \int_{x_2-i(x_2) \in G }p_2(x_2|x_1) d x_2.
\end{split}
\end{align}

We can check for $\forall x_2^{'}, x_2^{''}$ and $x_1$ such that $p(x_1) \geq \frac{1}{N_0}$,
\begin{align}
|p_2(x_2^{'}|x_1)-p_2(x_2^{''}|x_1)|= |\frac{p(x_1,x_2^{'})-p(x_1,x_2^{''})}{p_1(x_1)}|\leq N_0|p(x_1,x_2^{'})-p(x_1,x_2^{''}) |.
\end{align}

By the result that a continuous function in a compact set is uniformly continuous, there exists a $M_1$, such that $|p(x_1,x_2^{'})-p(x_1,x_2^{''})|\leq \frac{\delta}{2N_0R}$, $\forall (x_1,x_2^{'}),(x_1,x_2^{''})\in B(0,R)$ and $|x_2^{'}-x_2^{''}|\leq \frac{1}{M_1}$. Consequently, using the similar arguments as in \cref{appx:proof_thm_prelim}, if $m\geq 2M_1$,
\begin{equation}
\begin{aligned}
|\int_{x_2-i(x_2) \in G }p_2(x_2|x_1) d x_2-\int_{x_2-i(x_2) \notin G }p_2(x_2|x_1) d x_2| 
& \leq 
\max_{|a-b|\leq \frac{1}{M_1}}|p_2(a|x_1)-p_2(b|x_1) |\times 2R \\
& \leq 
N_0 \frac{\delta}{2N_0R}\times 2R \\ 
& \leq \delta,
\end{aligned}
\end{equation}
which implies
\begin{align}
\int_{x_2-i(x_2) \in G }p_2(x_2|x_1) d x_2\geq \frac{1}{2}-\frac{\delta}{2},
\end{align}
$\forall x_1$ such that $p_1(x_1)\geq \frac{1}{N_0}$.
Therefore
\begin{align}
\begin{split}
  & \mathbf{P}(x_1-i(x_1) \in G,x_2-i(x_2)\in G)
 \\ \geq  & \mathbf{P}(x_1-i(x_1) \in G \bigcap \bigcup_{j=1}^{N_0}\{  
p_1(x_1)\geq \frac{1}{j}
\},x_2-i(x_2)\in G)  \\= & \int_{x_1-i(x_1)\in G\bigcap \bigcup_{j=1}^{N_0} \{  
p_1(x_1)\geq \frac{1}{j}
\}} p_1(x_1)d x_1 \int_{x_2-i(x_2) \in G }p_2(x_2|x_1) d x_2
\\ \geq & \int_{x_1-i(x_1)\in G\bigcap \bigcup_{j=1}^{N_0} \{  
p_1(x_1) \geq \frac{1}{j}
\}} p_1(x_1)d x_1 (\frac{1}{2}-\frac{\delta}{2}) \\
\geq & (\frac{1}{2}-\delta)(\frac{1}{2}-\frac{\delta}{2})\geq \frac{1}{4}-\delta
\end{split}
\end{align}
given that $\delta<\frac{1}{2}$ when $m>\max\{M_0,M_1\}$.
Since we could choose $\delta$ arbitrarily, we have
\begin{align}
\label{sum_1}
\liminf_{m\to \infty}\mathbf{P}(x_1-i(x_1) \in G,x_2-i(x_2)\in G) \geq \frac{1}{4},
\end{align}
similarly, we have
\begin{align}
\label{sum_2}
\liminf_{m\to \infty}\mathbf{P}(x_1-i(x_1) \in G,x_2-i(x_2)\notin G) \geq \frac{1}{4},
\end{align}
\begin{align}
\label{sum_3}
\liminf_{m\to \infty}\mathbf{P}(x_1-i(x_1) \notin G,x_2-i(x_2)\in G) \geq \frac{1}{4},
\end{align}
\begin{align}
\label{sum_4}
\liminf_{m\to \infty}\mathbf{P}(x_1-i(x_1) \notin G,x_2-i(x_2)\notin G) \geq \frac{1}{4}.
\end{align}
However, the sum of \eqref{sum_1} to \eqref{sum_4} is 1. This implies that the limitations in \eqref{sum_1} to \eqref{sum_4} are all $\frac{1}{4}$, which finishes the proof.
\end{proof}

Now we are ready to prove \cref{thm:preliminary_p_dimensional}.
\begin{proof}
\begin{align}
\mathbf{P}(\bigcap_{k=1}^{p} A_k) &=\mathbf{P}(\|x\|_2 \leq R)\mathbf{P}(\bigcap_{k=1}^{p}A_k | \|x \|_2 \leq R)+\mathbf{P}(\| x\|_2 > R)\mathbf{P}(\bigcap_{k=1}^{p}A_k| | \|x\|_2 > R).
\end{align}

We can choose $R$ such that 
\begin{align}
|\mathbf{P}(\|x\|_2 \leq R)-1|<\epsilon.
\end{align}

Also by \cref{thm:preliminary_p_bounded_version}, there exists $M$, such that when $m>M$,
\begin{align}
|\mathbf{P}(\bigcap_{k=1}^{p}A_k | \|x \|_2 \leq R)-\frac{1}{2}|<\epsilon.
\end{align}

Therefore when $m>M$,
\begin{equation}
\begin{aligned}
|\mathbf{P}(\bigcap_{k=1}^{p} A_k)-\frac{1}{2}|
& \leq |\mathbf{P}(\|x\|_2 \leq R)| |\mathbf{P}(\bigcap_{k=1}^{p}A_k | \|x \|_2 \leq R)-\frac{1}{2}|+\mathbf{P}(\|x\|_2 > R)(1-\frac{1}{2}) \\ 
& \leq \epsilon+\frac{1}{2}\epsilon<2\epsilon.
\end{aligned}
\end{equation}
Since $\epsilon$ is arbitrary, we know that the convergence holds.
\end{proof}

\subsection{Proof of \cref{thm:high_dimensional}}
To prove this result, we consider a lemma from \cite{bentkus2005lyapunov}.
\begin{lemma}
\label{first_result_convex}
(c.f. Theorem 1.1 in \cite{bentkus2005lyapunov}) Let $\mathbf{y}_1, \ldots, \mathbf{y}_n$ be independent $p$-dimensional random vectors with a common mean $\mathbf{E} \mathbf{y}_j=0$. Write $S_Y=\mathbf{y}_1+\cdots+\mathbf{y}_n$. Throughout we assume that $S_Y$ has a nondegenerated distribution in the sense that the covariance operator, say $C^2=\operatorname{Cov} S_Y$, is invertible ( $C$ stands for the positive root of $\left.C^2\right)$. Let $Z$ be a Gaussian random vector such that $\mathrm{E} Z=0$ and $\operatorname{Cov} S_Y$ and $\operatorname{Cov} Z$ are equal. Write
$$
\beta=\beta_1+\cdots+\beta_n, \quad \beta_k=\mathbf{E}\left|C^{-1} \mathbf{y}_k\right|_2^3,
$$
and
$$
\Delta(\mathcal{C})=\sup _{A \in \mathcal{C}}|\mathbf{P}\{S_Y \in A\}-\mathbf{P}\{Z \in A\}|
$$
where $\mathcal{C}$ stands for the class of all convex subsets of $\mathbf{R}^p$. Then there exists an absolute positive constant $c$, such that
$$
\Delta(\mathcal{C}) \leq  c p^{1 / 4} \beta
$$
\end{lemma}
Now we are ready to prove 
\cref{thm:high_dimensional}.
\begin{proof}
Note that 
\begin{align}
(T_1-\frac{n}{2},T_2-\frac{n}{2},\cdots,T_p-\frac{n}{2}) \overset{d}{=}\sum_{j=1}^{n}\mathbf{y}_j,
\end{align}
where $\{ \mathbf{y}_j,j=1,2,\cdots n\}$ are i.i.d. random vectors, with each of them having independent components with mean $0$ and variance $\frac{1}{4}$.
Then using the same notations as in \cref{first_result_convex} and the monotonicity of $l_p$ norm,
\begin{align}
\beta_k=\mathbf{E}\left|C^{-1} \mathbf{y}_k\right|_2^3\leq \left\{\mathbf{E}\left|C^{-1} \mathbf{y}_k\right|_2^2 \right\}^{\frac{3}{2}}=8p^{\frac{3}{2}}\frac{1}{n^{\frac{3}{2}}},
\end{align}
and 
\begin{align}
\beta=\sum_{j=1}^{n}\beta_j\leq 8\frac{p^{\frac{3}{2}}}{n^{\frac{1}{2}}},
\end{align}
therefore
\begin{align}
\Delta(\mathcal{C})=\sup _{A \in \mathcal{C}}|\mathbf{P}\{S_Y \in A\}-\mathbf{P}\{Z \in A\}| \leq 8 c \frac{p^{\frac{7}{4}}}{n^{\frac{1}{2}}},
\end{align}
which implies if $p=o(n^{\frac{2}{7}})$, 
$\Delta(\mathcal{C}) \to 0$ as $n\to \infty$. Note that $\mathcal{C}$ stands for the class of all convex subsets of $\mathbf{R}^p$, we further have
\begin{align}
\sup _{r \geq 0}|\mathbf{P}\{ \| 2{n}^{-\frac{1}{2}} S_Y\|_2 \leq r\}-\mathbf{P}\{  \|2{n}^{-\frac{1}{2}}  Z \|_2 \leq r\}| \to 0,
\end{align}
which further implies
\begin{align}
\sup _{r \geq 0}|\mathbf{P}\{ \| 2{n}^{-\frac{1}{2}} S_Y\|^2_2 \leq r\}-\mathbf{P}\{  \|2{n}^{-\frac{1}{2}}  Z \|^2_2 \leq r\}| \to 0,
\end{align}
note that
$\| 2{n}^{-\frac{1}{2}} S_Y\|^2_2=\sum_{j=1}^{p} \left[2\sqrt{n} \left(\frac{T_j}{n} - \frac{1}{2}\right)\right]^2$
and 
$ \|2{n}^{-\frac{1}{2}}  Z \|^2_2 \sim \chi^2_p$, we finish the proof.
\end{proof}


\subsection{Proof of \cref{thm:robustness}}
\label{appx:proof_thm_robust}
\begin{proof}
Without loss of generality, we assume that the probability $x_{ij}$ is attacked successfully is $\frac{1}{2}$, $\forall i,j$, we use $k_i$ to denote the number of elements moved out of the green list in $ith$ column,
then according to Hoeffding’s inequality (c.f. Theorem 2.2.6 in \cite{vershynin2018high}),
\begin{align}
\mathbf{P}(\sum_{j=1}^{p}k_j \leq (1+\delta)  (\frac{1}{2}\sum_{j=1}^{p}\widehat{k}_j ))\geq 1-e^{-\frac{2 \delta^2 (\sum_{j=1}^{p}\frac{1}{2}\widehat{k}_j)^2 }{\sum_{j=1}^{p}\widehat{k}_j}}= 1-e^{-\frac{1}{2}\delta^2  \sum_{j=1}^{p}\widehat{k}_j},
\end{align}
$\forall \delta>0$.
Take $\delta=\frac{1}{{(np)}^{\frac{1}{4}}}$, we will have as long as 
\begin{align}
 \sum_{j=1}^{p}\widehat{k}_j \leq 
 \frac{1}{1+\frac{1}{(np)^{\frac{1}{4}}}}(np-\sqrt{np}\sqrt{\mathcal{X}_p^2(1-\alpha)}),
 \end{align}
 with a probability at least 
 \begin{align}
1-e^{-\frac{1}{2}(\sqrt{np}-\sqrt{\mathcal{X}_p^2(1-\alpha)})},
 \end{align}
 we have
 \begin{align}
 \sum_{j=1}^{p}k_j \leq \frac{1}{2}(np-\sqrt{np}\sqrt{\mathcal{X}_p^2(1-\alpha)}),
 \end{align}
 and therefore the attack fails.
\end{proof}

\newpage
\section{Python-Style Pseudo-Code}
\label{appx:pseudo_code}
We provide python-style pseudocode to facilitate understanding of the proposed watermark. During testing, one can count the number of elements in the $j$-th column falling inside the ``green list'' intervals as $t_j$, and perform binomial or chi-square hypothesis-testing.

\begin{lstlisting}[language=Python, caption={Tabular data watermark.}]
import numpy as np

def getGreenList(lo=0, hi=1, m=1000):
    """ return a list of tuple, representing the green list intervals
    """

    waymarks = np.linspace(lo, hi, m + 1)

    green_list = []
    for i in range(0, m, 2):
        # randomly select one interval from each pair
        if np.random.uniform() > .5:
            i += 1
        green_list.append([waymarks[i], waymarks[i + 1]])

    return green_list

""" 
    Watermarking a p-column table:
    step 1: generate p green lists
    step 2: for each column, call 
            `singleColumnWatermark(arr, green_list)`
            to watermark the column vector.
"""

def singleColumnWatermark(arr, green_list):
    arr_wm = arr.copy()

    for i in range(arr_wm):
        # offset elem to [0, 1]
        e_flr = np.floor(arr_wm[i])
        e = arr_wm[i] - e_flr

        # find the nearest interval in the ``green list'' intervals
        g = findNearestInterval(e, green_list)

        if e > g[1] or e < g[0]:
            # if x[i] falls outside of the range, then
            # we re-sample the elem. from a uniform dist.
            arr_wm[i] = np.random.uniform(g[0], g[1]) + e_flr

    return arr_wm
\end{lstlisting}

\newpage
\begin{lstlisting}[language=Python, caption={Finding nearest interval with O(1) complexity.}]
def findNearestInterval(e, green_list, m):
    """ return the nearest interval to the given element e
    """

    min_dist, min_indx = np.inf, -1

    # offset elem to [0, 1]
    e = e - np.floor(e)

    # find the nearest pair of intervals
    idx_c = int(e // (2 / m))
    # neighboring indices
    idx_l0, idx_r0 = max(0, idx_c - 1), \
                     min(idx_c + 1, len(green_list) - 1)
    idx_l1, idx_r1 = max(0, idx_c - 2), \
                     min(idx_c + 2, len(green_list) - 1)

    # local green lists with possible candidates 
    # including the closest interval
    local_g_list = [green_list[idx_l1], green_list[idx_l0], 
                    green_list[idx_c],  green_list[idx_r0], 
                    green_list[idx_r1]]
    for i, intv in enumerate(local_g_list):
        cur_dist = np.abs(e - (intv[0] + intv[1]) / 2)

        if cur_dist < min_dist:
            min_dist = cur_dist
            min_indx = i

    return local_g_list[min_indx]

\end{lstlisting}

\newpage
\section{Experiment Settings}
\label{appx:exp_set}

\subsection{Synthetic Datasets}
\label{appx:syn_dat}
For synthetic datasets, we set $m=1000$, which we find sufficient for all the experiments. For watermark detection without additive noise attacks, we vary the row and column number of the tabular data within the range $\{10, 100, 1000\}$, resulting in tables of sizes within the range of $\{10, 100, 1000\} \times \{10, 100, 1000\}$. For each set-up of the tabular data size, we create $1000$ watermarked and unwatermarked tables to calculate the detection rate (true positive rate) and specificity (true negative rate) with the significance level of $\alpha = .005$; the tabular data is from standard zero-mean multivariate Gaussian distribution. The specificity in all settings remain 1. We repeat the experiments for 3 independent runs to calculate the error bars in \cref{subfig:chi_sing_col,subfig:chi_mult_col}.

For watermark detection with additive noise attacks, we set the tabular data size as $5000 \times 100$ as a representative and vary the variance of additive zero-mean Gaussian noise within the range $\{0.001, 0.01, 0.1, 1, 10\}$. To verify our results in \cref{thm:robustness}, we vary the proportion of elements in the table being modified \texttt{prp} within the range $\{0.50, 0.75, 0.90, 0.95\}$. Similarly, we create $1000$ watermarked and unwatermarked tables to calculate the detection rate (true positive rate) and specificity (true negative rate) with $\alpha = .005$. The specificity in all settings remain 1. We repeat the experiments for 3 independent runs to calculate the error bars in \cref{subfig:chi_sing_col_atck,subfig:chi_mult_col_atck}.

\subsection{Real-World Datasets}
\label{appx:real_dat}
\paragraph{Dataset} The full list of datasets and their properties are presented in \cref{tab:dst_info}.

\begin{table}[!htbp]
    \vspace{-5pt}
    \centering
    \caption{\textbf{List of datasets used for the evaluation and their descriptions.} }
    \vskip 0.1in
    \begin{tabular}{cccccccc}
    \toprule
        Alias & Name & \#Train & \#Validation 
           & \#Test & \#Num & \#Cat & Task type \\
    \midrule
    California & California Housing & 13209 & 3303 & 4128 & 8 & 0 & Regression \\
    Gesture & Gesture Phase & 6318 & 1580 & 1975 & 32 & 0 & Multiclass \\
    House & House 16H & 14581 & 3646 & 4557 & 16 & 0 & Regression \\
    Wilt & Wilt & 3096 & 775 & 968 & 5 & 0 & Binclass \\
    Higg-small & Higgs Small & 62751 & 15688 & 19610 & 28 & 0 & Binclass \\
    Miniboone & MiniBooNE & 83240 & 20811 & 26013 & 50 & 0 & Binclass \\
    \bottomrule
    \end{tabular}
    \label{tab:dst_info}
    \vspace{-5pt}
\end{table}

\paragraph{Evaluation measure}
To investigate the performance of our tabular watermark on real-world data, we sample from each generative model a generated dataset with the size of a real training set as in \cref{tab:dst_info}. For each set-up in evaluating fidelity, accuracy and robustness, we create $50$ watermarked and unwatermarked training sets (\ie, $n \times p$ tables) to measure the wasserstein distance and ROC-AUC scores. For utility evaluation, we create $50$ watermarked and unwatermarked training sets to train CatBoost models \cite{prokhorenkova2018catboost} for classification and regression tasks, which are then evaluated on the real testing sets. In our experiments, classification performances are evaluated by the F1 score, and regression performance is evaluated by the R2 score.

\paragraph{Tunning process of tab. generators} 
We follow \cite{kotelnikov2023tabddpm} and use the Optuna library \cite{akiba2019optuna} to tune the hyperparameters of the tabular data generators. The tuning process is guided by the values of the ML efficiency (with respect to Catboost) of the generated synthetic data on a hold-out validation dataset (the score is averaged over five different sampling seeds). We refer to \cite{kotelnikov2023tabddpm} for search spaces for all hyperparameters of the tab. generators. We run the experiments on a A6000 GPU. Training and evaluation process typically finishes within 24 hrs.

\newpage
\section{Additional Experiment Results}
\label{appx:add_exp}

\paragraph{Additional results on simulated tables} We additionally provide the ROC-AUC scores on simulated results for watermarking a single column corresponding with \cref{subfig:chi_sing_col} in \cref{tab:simu_tabs}. The ROC-AUC scores for watermarking all columns are 1.
\begin{table}[!htbp]
    \vspace{-5pt}
    \centering
    \caption{\textbf{Detection results on watermarking a single column in simulated tables.}}
    \vskip 0.1in
    \begin{tabular}{cccc}
    \toprule
    $n \times p$ 
       & $10 \times 10$ 
       & $10 \times 100$ 
       & $10 \times 1000$ \\
    \midrule
    {AUC} & $0.850$ & $0.700$ & $0.580$ \\
    \midrule
    {$n \times p$ } & $100 \times 10$ 
       & $100 \times 100$
       & $100 \times 1000$ \\
    \midrule
    {AUC} & $1.000$ & $1.000$ & $0.970$ \\
    \midrule
    {$n \times p$ } & $1000 \times 10$ 
       & $1000 \times 100$
       & $1000 \times 1000$ \\
    \midrule
    {AUC} & $1.000$ & $1.000$ & $1.000$ \\
    \bottomrule
    \end{tabular}
    \label{tab:simu_tabs}
\end{table}

\paragraph{Results on simulated high-dim. tables} We provide simulation results on high dimensional tables, where the number of columns $p$ exceeds the row numbers $n$ in \cref{tab:high_dim_tabs}. The tabular data is from standard Gaussian. We observe similar results with tabular data from 5-component randomly initialized gaussian mixture models, which mimic multimodal distributions.
\begin{table}[!htbp]
    \vspace{-5pt}
    \centering
    \caption{\textbf{Detection results on simulated high dimensional tables.} We report the true postive (TPR) and true negative rates (TNR) as well as ROC-AUC scores. We create $100$ watermarked and unwatermarked tables to calculate the scores.}
    \vskip 0.1in
    \begin{tabular}{cccc}
    \toprule
        $n \times p$ & $100 \times 100$ & $100 \times 1000$ & $100 \times 10000$ \\
    \midrule
    TPR/TNR & $1.000/1.000$ & $1.000/1.000$ & $1.000/1.000$ \\
    AUC & $1.000$ & $1.000$ & $1.000$ \\
    \bottomrule
    \end{tabular}
    \label{tab:high_dim_tabs}
    \vspace{-5pt}
\end{table}

\paragraph{Additional results on simulated attacks} We additionally provide the ROC-AUC scores after simulated attacks corresponding with watermarking a single column (\cref{subfig:chi_sing_col_atck}) in \cref{tab:simu_atcks}. The ROC-AUC scores corresponding with watermarking all columns are 1.
\begin{table}[!htbp]
    \vspace{-5pt}
    \centering
    \caption{\textbf{Detection results after attacks on watermarking a single column in simulated tables.} \texttt{s} represents the variances of additive noises; \texttt{p} represents the proportions of elements being modified by the attacks.}
    \vskip 0.1in
    \begin{tabular}{ccccc}
    \toprule
    {} & $s = 0.001$ 
       & $s = 0.1$ 
       & $s = 1$ 
       & $s = 10$ \\
    \midrule
    {$p=0.50$} & $1.000$ & $1.000$ & $1.000$ & $1.000$ \\
    {$p=0.75$} & $1.000$ & $1.000$ & $1.000$ & $1.000$ \\
    {$p=0.90$} & $1.000$ & $0.970$ & $0.990$ & $0.870$ \\
    {$p=0.95$} & $1.000$ & $0.780$ & $0.850$ & $0.690$ \\
    \bottomrule
    \end{tabular}
    \label{tab:simu_atcks}
\end{table}

\paragraph{Results on Column Selection} We provide number of columns selected in \cref{tab:col_sel} for each generated dataset before and after applying our heuristic smoothness check method.
\begin{table}[!htbp]
    \vspace{-5pt}
    \centering
    \caption{\textbf{Number of columns selected for each generated dataset.} }
    \vskip 0.1in
    \begin{tabular}{ccccccc}
    \toprule
        {} & California & Gesture & House 
           & Wilt & Higgs-small & Miniboone \\
    \midrule
    TVAE & 8/8 & 32/32 & 16/16 & 5/5 & 24/28 & 50/50 \\
    CTABGAN & 8/8 & 32/32 & 16/16 & 5/5 & 24/28 & 50/50 \\
    TabDDPM & 5/8 & 32/32 &  7/16 & 5/5 & 24/28 & 27/50 \\
    \bottomrule
    \end{tabular}
    \label{tab:col_sel}
    \vspace{-5pt}
\end{table}

\paragraph{Additional results on data fidelity} We provide additional results of data fidelity for TVAE and CTABGAN, shown in \cref{tab:wdist_tvae,tab:wdist_ctabgan}.
\begin{table}[!htbp]
    \vspace{-5pt}
    \centering
    \caption{\textbf{Wasserstein-1 distance between TVAE generated data and watermarked data distributions.}}
    \vskip 0.1in
    \begin{tabular}{ccccccc}
    \toprule
        {} & California & Gesture & House 
           & Wilt & Higgs-small & Miniboone \\
    \midrule
    Orig2Gen & 0.1327 & 0.1198 & 0.0658
             & 0.0906 & 0.1797 & 0.6720 \\
    Gen2Watermarked 
             & 0.0004 & 0.0004 & 0.0004 
             & 0.0003 & 0.0003 & 0.0004 \\
    \midrule
    $m$      & 1000   & 1000   & 1000   
             & 1000   & 1000   & 1000   \\
    $1 / m$  & 0.001 & 0.001   & 0.001  
             & 0.001 & 0.001   & 0.001 \\
    \bottomrule
    \end{tabular}
    \label{tab:wdist_tvae}
    \vspace{-5pt}
\end{table}
\begin{table}[!htbp]
    \vspace{-5pt}
    \centering
    \caption{\textbf{Wasserstein-1 distance between CTABGAN generated data and watermarked data distributions.}}
    \vskip 0.1in
    \begin{tabular}{ccccccc}
    \toprule
        {} & California & Gesture & House 
           & Wilt & Higgs-small & Miniboone \\
    \midrule
    Orig2Gen & 0.1683 & 0.1771 & 0.0926
             & 0.1003 & 0.0606 & 0.6471 \\
    Gen2Watermarked 
             & 0.0004 & 0.0004 & 0.0004 
             & 0.0004 & 0.0004 & 0.0004 \\
    \midrule
    $m$      & 1000   & 1000   & 1000   
             & 1000   & 1000   & 2500   \\
    $1 / m$  & 0.001 & 0.001   & 0.001  
             & 0.001 & 0.001   & 0.0004 \\
    \bottomrule
    \end{tabular}
    \label{tab:wdist_ctabgan}
    \vspace{-5pt}
\end{table}

\paragraph{Examples of ill-shaped column data distributions} We provide examples of ill-shaped column distributions in \cref{fig:appx_ill_dist}. We can see these distributions have spikes concentrated on certain values (especially in the left subfig), with undesirable smoothness violating our assumption in \cref{thm:preliminary}.
\begin{figure}[!h]
\centering
\begin{minipage}[t]{.49\textwidth}
    \centering
    \includegraphics[scale=0.45]{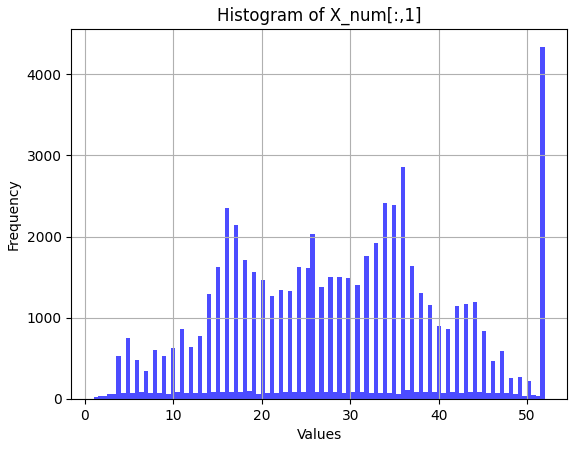}
\end{minipage}%
\hfill
\begin{minipage}[t]{.49\textwidth}
    \centering
    \includegraphics[scale=0.45]{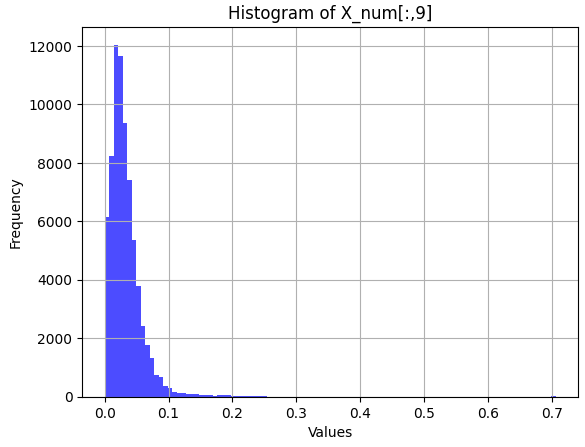}
\end{minipage}
\captionof{figure}{Histogram of spiky column data distributions. Examples generated by TabDDPM.}
\label{fig:appx_ill_dist}
\end{figure}

\newpage
\section{Limitations and Future work}
\label{appx:limit}
One potential limitation of our framework is that it primarily addresses continuous variables, leaving the watermarking of discrete variables as an area for future exploration. It would be worthwhile to investigate whether a similar technique involving ``green list'' and ``red list'' intervals can be effectively applied to discrete variables. Additionally, the specific choice of $m$
(the number of ``green list'' intervals) is closely tied to the smoothness of the data distribution. For instance, if the distribution under the null hypothesis is not smooth and exhibits characteristics such as spikes, a larger $m$ would be necessary to ensure that the probability of a sample point falling within a ``green list'' interval approaches $\frac{1}{2}$. These aspects highlight the need for further refinement and adaptation of our framework to accommodate a broader range of data types and distributional properties.

\section{Broader Impacts}
\label{appx:broader_impct}
This work contributes to enhancing the security of both synthetic and real-world datasets, which is critical in the context of AI and machine learning applications.
Specifically, generative models could be misused for disinformation or faking profiles. Our work focuses on watermarking generative data, which can facilitate the reliable detection of synthetic content, and could be important to address such harms from generative models.
We consider our work to be foundational and not tied to particular applications or deployments. It is possible that future works may involve malicious uses of this technique that we are unaware of for now.

\end{document}